\documentclass[journal]{IEEEtran}  

\usepackage[utf8]{inputenc} 
\usepackage[T1]{fontenc}
\usepackage{url}
\usepackage{ifthen}
\usepackage{cite}
\usepackage[cmex10]{amsmath} 
\usepackage{amssymb,mathtools}
\usepackage{graphicx}
\usepackage{multirow}
\usepackage{amsthm}
\usepackage{cuted}

\usepackage{array}
\newcolumntype{M}[1]{>{\centering\arraybackslash}m{#1}}
\newcolumntype{N}{@{}m{0pt}@{}}
\newcommand\Tstrut{\rule{0pt}{5ex}}       
\newcommand\Bstrut{\rule[-0.9ex]{0pt}{0pt}} 
\newcommand{\TBstrut}{\Tstrut\Bstrut} 

\usepackage{framed}
\usepackage[framemethod=TikZ]{mdframed}

\usepackage{enumerate}

\usepackage{tikz}
\usetikzlibrary{shapes,snakes}
\usetikzlibrary{tikzmark}
\usetikzlibrary{fit,backgrounds}
\usetikzlibrary{matrix,decorations.pathreplacing,angles,quotes}
\usetikzlibrary{patterns}
\usetikzlibrary{calc}
\usetikzlibrary{positioning}
\tikzstyle{block} = [rectangle, draw, 
    minimum width=4em, text centered, rounded corners, minimum height=1em]
\tikzstyle{rec} = [rectangle, draw]
\tikzstyle{line} = [draw, -latex]
\usetikzlibrary{intersections}

\usepackage{color}

\newtheorem{lemm}{Lemma}[section]

\newtheorem{theo}{Theorem}[section]
\newtheorem{prop}{Proposition}[section]

\theoremstyle{definition}
\newtheorem{remark}{Remark}[section]

\DeclareMathOperator{\Tr}{\mathrm{Tr}}

\DeclareMathOperator{\cA}{{\mathcal{A}}}
\DeclareMathOperator{\cB}{{\mathcal{B}}}

\DeclarePairedDelimiterX\bkk[2]{\langle}{\rangle}{#1 \delimsize\vert #2}
\DeclarePairedDelimiterX\bk[2]{\langle}{\rangle}{#1 \delimsize\vert #1}
\DeclarePairedDelimiterX\kbb[2]{\vert}{\vert}{#1 \rangle\langle #2}
\DeclarePairedDelimiterX\kb[1]{\vert}{\vert}{#1 \rangle\langle #1}

\DeclarePairedDelimiter\paren{(}{)}

\newcommand\vN{\mathsf{n}}
\newcommand\vF{\mathsf{f}}
\newcommand\vM{\mathsf{m}}

\newcommand\vR{\mathsf{r}}
\newcommand\vkc{\mathsf{k}}

\begin{document}

\title{Capacity of Quantum Private Information Retrieval with Multiple Servers} 

\author{%
Seunghoan Song,~\IEEEmembership{Student Member,~IEEE}, 
 and~Masahito~Hayashi,~\IEEEmembership{Fellow,~IEEE}
\thanks{This article was presented in part at Proceedings of 
2019 IEEE International Symposium on Information Theory \cite{SH19}.}
\thanks{S. Song is with Graduate school of Mathematics, Nagoya University, Nagoya, 464-8602, Japan
(e-mail: m17021a@math.nagoya-u.ac.jp).}
\thanks{M. Hayashi is with 
Shenzhen Institute for Quantum Science and Engineering, Southern University of Science and Technology,
Shenzhen, 518055, China,
Guangdong Provincial Key Laboratory of Quantum Science and Engineering,
Southern University of Science and Technology, Shenzhen 518055, China,
Shenzhen Key Laboratory of Quantum Science and Engineering, Southern
University of Science and Technology, Shenzhen 518055, China,
and Graduate School of Mathematics, Nagoya University, Nagoya, 464-8602, Japan
(e-mail:hayashi@sustech.edu.cn).}
\thanks{
SS is grateful to Dr. Hsuan-Yin Lin for helpful discussions and comments.
SS is supported by Rotary Yoneyama Memorial Master Course Scholarship (YM), Lotte Foundation Scholarship, and JSPS Grant-in-Aid for JSPS Fellows No. JP20J11484. 
MH is supported in part by Guangdong Provincial Key Laboratory (Grant No. 2019B121203002),
a JSPS Grant-in-Aids for Scientific Research (A) No.17H01280 and for Scientific Research (B) No.16KT0017, and Kayamori Foundation of Information Science Advancement.
}
\thanks{Copyright (c) 2017 IEEE. Personal use of this material is permitted.  However, permission to use this material for any other purposes must be obtained from the IEEE by sending a request to pubs-permissions@ieee.org.}
}

\maketitle

\begin{abstract}

%
%

We study the capacity of quantum private information retrieval (QPIR) with multiple servers.
In the QPIR problem with multiple servers, a user retrieves a classical file by downloading quantum systems from multiple servers 	
	each of which contains the copy of a classical file set while the identity of the downloaded file is not leaked to each server.
The QPIR capacity is defined as 
the maximum rate of the file size over the whole dimension of the downloaded quantum systems.
When the servers are assumed to share prior entanglement, 
	we prove that the QPIR capacity with multiple servers is $1$
regardless of the number of servers and files.
We construct a rate-one protocol only with two servers.
This capacity-achieving protocol outperforms its classical counterpart in the sense of the capacity, server secrecy, and upload cost.
The strong converse bound is derived concisely without using any secrecy condition.
We also prove that the capacity of multi-round QPIR is $1$.
\end{abstract}

\section{Introduction}

Introduced by the seminal paper \cite{CGKS98}, 
Private Information Retrieval (PIR) finds efficient methods to download a file from non-communicating servers 
	each of which contains the copy of a classical file set while the identity of the downloaded file is not leaked to each server.
This problem is trivially solved by requesting all files to one of the servers, but this method is inefficient.
Finding an efficient method is the goal of this problem and it has been extensively studied in many papers \cite{CMS99, Lipmaa10, BS03, DGH12}.
Moreover,  the papers \cite{KdW03,KdW04, Ole11,BB15, LeG12, KLLGR16, ABCGLS19} have studied the Quantum PIR (QPIR) problem
where the user downloads quantum systems instead of classical bits to retrieve a classical file from the servers.

In classical PIR studies,
    the paper \cite{SJ16} started the discussion of the PIR capacity. 
The PIR capacity is defined by the maximum rate of the file size over the download size when the numbers of the servers and the files are fixed.
In the definition of the PIR capacity, the upload cost, i.e., the total size of the queries, is neglected since it does not scale with the file size which is allowed to go infinity.
	When each of the $\vN$ servers contains a copy of the $\vF$ files,
    the paper \cite{SJ16} showed that 
    the PIR capacity is $(1-1/\mathsf{n})/(1-(1/\vN)^{\vF})$.
Moreover, the paper \cite{TSC18} proposed a capacity-achieving protocol with the minimum upload cost and minimum file size in a class of PIR protocols.
{Furthermore, after \cite{SJ16}, several PIR capacities have been studied under different problem settings.
Symmetric PIR is the PIR with server secrecy that the user obtains no more information than the target file,
and the capacity of the symmetric PIR is $1-{\vN}^{-1}$ \cite{SJ16-2}.
Another extension is
the PIR with coded databases \cite{CHY15, FHGHK17, BU18,KLRG17,LKRG18,HFH19}, where the files are coded and distributed to the servers. 
When the files are coded by an $(\vN,\vkc)$ Maximum Distance Separable (MDS) code, 
	the PIR capacity is $(1-\vkc/\vN)/(1-\paren*{\vkc/\vN}^{\vF})$ \cite{BU18}.
Multi-round PIR has also been studied in \cite{SJ18} and the capacity was proved to be the same as the PIR capacity derived in \cite{SJ16}.
}


    On the other hand, the QPIR problem is rarely treated with multiple servers and there is no study on the capacity of the QPIR problem.
    Though the papers \cite{KdW03,KdW04} treated the QPIR problem with multiple servers,
    they evaluated the communication complexity which is the sum of upload and download costs required to retrieve a one-bit file
    instead of the capacity. 

\begin{table}[t]   \label{tab:compare}
\renewcommand{\arraystretch}{1.6}
\begin{center}
\caption[caption]{Capacities of classical and quantum PIRs}
{
\setlength\extrarowheight{0.0em}
\begin{tabular}{|M{7em}|M{9.5em}|M{9.5em}|N}
\cline{1-3}
                & Classical PIR Capacity & Quantum PIR Capacity  & \\
                \cline{1-3}
PIR             & $\displaystyle \frac{1-{\vN}^{-1}}{1-{\vN}^{-\vF}}$ \cite{SJ16}& $1$ ${}^{\ddagger}$  & \TBstrut \\[0.9em]
                \cline{1-3}
Symmetric PIR   & $\displaystyle 1-{\vN}^{-1}$ \cite{SJ16-2}& $1$ ${}^{\ddagger}$ &  \TBstrut\\[0.6em]
                \cline{1-3}
Multi-round PIR 
                & $\displaystyle \frac{1-{\vN}^{-1}}{1-{\vN}^{-\vF}}$ \cite{SJ18} &    $1$  &\TBstrut\\[0.9em]
                \cline{1-3}
\multicolumn{3}{l}{\footnotesize $\ast$ $\vN$, $\vF$: the numbers of servers and files, respectively.}\\[-0.6em]
\multicolumn{3}{l}{\footnotesize $\dagger$ Shared randomness among servers is necessary.}\\[-0.6em]
\multicolumn{3}{l}{\footnotesize ${\ddagger}$ Capacities are derived with the strong converse bounds.}\\
\end{tabular}
}
\end{center}
\end{table}

    In this paper, as quantum extensions of the classical PIR capacities \cite{SJ16, SJ16-2, SJ18},
    we show that the capacities of QPIR, symmetric QPIR, and multi-round QPIR are $1$.
    We derive the QPIR capacity when a user retrieves a file secretly from non-communicating $\vN$ servers each of which contains the whole set of $\vF$ files by downloading quantum states 
    under the assumption that an entangled state is shared previously among all servers.
    We evaluate the security of a QPIR protocol with three parameters: the retrieval error probability, the user secrecy that the identity of the queried file is unknown to any individual server, and the server secrecy that the user obtains no more information than the target file. 
As a main result, 
    we show that the QPIR capacity is $1$ 
    regardless of whether it is of exact/asymptotic security and with/without the restriction that the upload cost is negligible to the download cost.
We propose a rate-one QPIR protocol with perfect security and finite upload cost.
We prove the converse bound that the rate of any QPIR protocol is less than $1$ even with no secrecy, no upload constraint, and any error probability.
    Moreover, we show that the capacity of multi-round QPIR is $1$. 
    We prove the weak converse bound of the multi-round QPIR capacity, i.e., the upper bound when the error probability is asymptotically zero.

\begin{table}[t]   \label{tab:protocol}
\renewcommand{\arraystretch}{1.6}
\begin{center}
\caption[caption]{Comparison of protocols in this paper and \cite{TSC18} }
\begin{tabular}{|M{9em}|M{8em}|M{9em}|N}
\cline{1-3}
                & This paper & Paper \cite{TSC18}  & \\
                \cline{1-3}
Server secrecy  & Yes & No  \\
                \cline{1-3}
Capacity  & $1$ & $\!(1\!-\!{\vN}^{-1})/(1\!-\!{\vN}^{-\vF})\!$   &\\
                \cline{1-3}
Condition for capacity $1$
                &  $\vN\geq 2$     &  $\vN\to \infty$  &\\
                \cline{1-3}
Upload cost
                & $2\vF$ bits &   
                $\displaystyle \vN(\vF-1)\log {\vN}$ bits &\TBstrut\\
                \cline{1-3}
Possible file sizes
&
$\{\ell^2\}_{\ell=2}^{\infty}$ &  $\{\ell^{\vN-1}\}_{\ell=2}^{\infty}$   &\\
                \cline{1-3}
\multicolumn{3}{l}{\footnotesize $\ast$ Server secrecy is the property that the user obtains no information}\\[-0.8em]
\multicolumn{3}{l}{\footnotesize \phantom{$\ast$} other than the target file.}\\[-0.6em]
\multicolumn{3}{l}{\footnotesize $\dagger$ $\vN$, $\vF$: the numbers of servers and files, respectively.}\\[-0.6em]
\multicolumn{3}{l}{\footnotesize $\ddagger$ Upload cost is the total bits which are sent to the servers.}\\
\end{tabular}
\end{center}
\end{table}


    Our capacity-achieving protocol has several remarkable advantages compared to the protocol \cite{TSC18} with the minimum upload cost and the minimum file size (see Table \ref{tab:protocol}).
    First, 
	our protocol is a symmetric QPIR protocol
    which guarantees the server secrecy,
    i.e.,
    the user obtains no information of files other than the retrieved one.
    This contrasts with the protocol in \cite{TSC18} that retrieves some information of the other files.
    Secondly, our protocol keeps the secrecy against the malicious user and servers.
    That is, 
    the user cannot obtain more information than the target file 
    even if the user sends malicious queries to the servers, 
    and 
    the servers cannot obtain the identity of the user's target file
    even if the servers answer malicious file information.
    Thirdly, the rate $1$ of our protocol is greater than the rate $(1-{\vN}^{-1})/(1-{\vN}^{-\vF})$ of the protocol in \cite{TSC18}.
    Fourthly, 
    our protocol achieves the capacity $1$ only with two servers. That is, in the sense of the QPIR capacity, there is no benefit to using more than two servers.
    On the other hand, in the protocol in \cite{TSC18}, 
    the capacity is strictly increasing in the number of servers and strictly decreasing in the number of files, 
    and an infinite number of servers are needed to achieve the capacity 1.
    Fifthly, our protocol needs the upload of 2$\vF$ bits
    whereas the protocol in \cite{TSC18} needs $(\vN(\vF-1)\log \vN)$-bit upload.
    Lastly, our protocol is defined when the file size $\vM$ is the square of any integer,
    but the protocol in \cite{TSC18} requires the file size $\vM$ to be the $(\vN-1)$-th power of any integer.
%

     The converse proofs of the QPIR capacities are much simpler than those of the PIR capacities \cite{SJ16,SJ16-2, SJ18}.
     Whereas the papers \cite{SJ16, SJ16-2, SJ18} used several entropic inequalities based on the assumptions on the PIR problem,
     our converse bounds are concisely derived without using any secrecy conditions but only by focusing on the download step of QPIR protocol.

    It should be noted that our QPIR protocol can be considered as a distributed version of Oblivious Transfer (OT) \cite{Rabin81, EGL85}.
    OT is equivalent to the symmetric PIR with one server and
    therefore, the symmetric PIR with multiple servers can be considered as a distributed version of OT.
    OT is an important cryptographic protocol because the free uses of an OT protocol construct an arbitrary secure multiparty computation \cite{Kilian88, IPS08}.
    Unfortunately, the symmetric classical PIR cannot be constructed without shared randomness among servers \cite{GIKM00}.
    On the other hand, the paper \cite{KdW04} showed that the two-way quantum communication between the user and the servers enables the symmetric PIR without shared randomness.
    Our result extends the result \cite{KdW04} so that 
    the symmetric PIR can also be constructed without shared randomness even for the case of classical upload, quantum download, and prior entanglement among the servers. 
    {Note that} if the quantum upload is allowed to our model, 
		we do not need the assumption of the prior entanglement because the user can upload an entangled state to all servers.

We assume that the servers share an entangled state before the protocol starts, but do not communicate. 
Similarly, the prior entanglement by non-communicating multiple parties has been assumed in another quantum protocol, called Multi-prover Quantum Interactive Proof (MQIP) \cite{KM03,RUV12,CHTW04}.
In MQIP, to solve a computational problem, a computationally-limited verifier sends queries to multiple provers who do not communicate with each other and have quantum computers, 
	receives answers from them, and then verifies whether the answers from the provers give a correct solution of the problem.
Similarly to our assumption, the MQIP studies \cite{KM03,RUV12,CHTW04} assumed that the provers share an entangled state before the protocol starts, but do not communicate with each other.
However, 
	even if the assumption of shared entanglement is similar, our communication model is different from those in \cite{KM03,RUV12,CHTW04}:
the papers \cite{KM03,RUV12} treated the quantum queries and quantum answers, and the paper \cite{CHTW04} treated the classical queries and classical answers,
but our paper treats classical queries and quantum answers.

The remaining of this paper is organized as follows.
Section~\ref{sec:problem} 
presents the formal definition of the QPIR protocol and the QPIR capacity 
and proposes the QPIR capacity theorem.
Section~\ref{sec:achieve} constructs the rate-one QPIR protocol
and {analyzes the security of our protocol against the malicious user and servers.}
Section~\ref{sec:converse} proves the converse bound.
{Section~\ref{sec:multi} extends the result to the capacity of multi-round QPIR.}
Section~\ref{sec:conclusion} is the conclusion of this paper.

\subsubsection*{Notations}
\textit{
For any set $\mathcal{T}$, we denote by $|\mathcal{T}|$ the cardinality of the set $\mathcal{T}$ and
by $\mathsf{I}_{\mathcal{T}}$ (or $\mathsf{I}$) the identity operator on $\mathcal{T}$.
For any matrix $\mathsf{B}$, we denote by $\bar{\mathsf{B}}$ the complex conjugate of  $\mathsf{B}$ and $\mathsf{B}^{\dagger}\coloneqq \bar{\mathsf{B}}^{\top}$.
The set of integers is denoted by $\mathbb{Z}$, and $\mathbb{Z}_d\coloneqq\mathbb{Z}/d\mathbb{Z}$ for any integer $d$.
$\Pr_X[A]$ denotes the probability that the variable $X$ satisfies the condition $A$.
For any quantum system $\cA$, we denote the set of all quantum states on $\cA$ by $\mathcal{S}(\cA)$.
}

\section{QPIR Protocol and Main Theorem} \label{sec:problem}

In this section, we formally define the QPIR protocol and its capacity
and presents the main theorem of the paper.
{For preliminaries on quantum information theory, see Appendix~\ref{append:qttheory}.}


\begin{small}
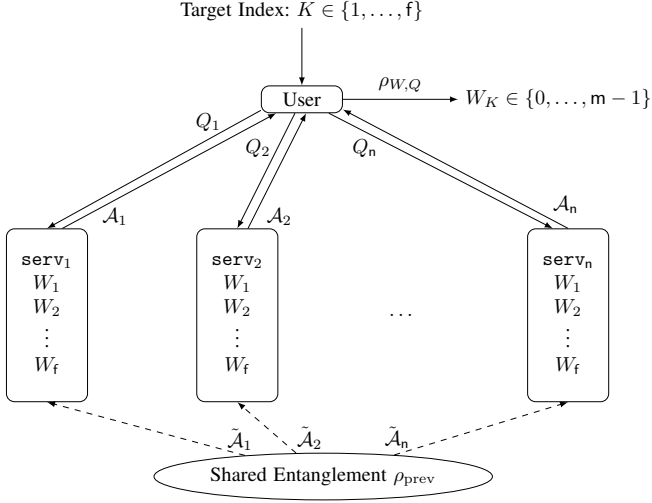
\begin{figure}[t]
\begin{center}
        \resizebox {1\linewidth} {!} {
\begin{tikzpicture}[scale=0.5, node distance = 3.3cm, every text node part/.style={align=center}, auto]
    \node [block] (user) {User};
    \node [above=1cm of user] (sent) {Target Index: $K\in\{1,\ldots,\vF\}$};
    \node [block,minimum height = 3cm, below left=2cm and 3cm of user] (serv1) {$\mathtt{serv}_1$\\$W_1$\\$W_2$\\\vdots\\$W_{\vF}$};
    \node [block,minimum height = 3cm, right of=serv1] (serv2) {$\mathtt{serv}_2$\\$W_1$\\$W_2$\\\vdots\\$W_{\vF}$};
    \node [right=1.8cm of serv2] (ten) {$\cdots$};
    \node [block,minimum height = 3cm, right=1.8cm of ten] (servn) {$\mathtt{serv}_{\vN}$\\$W_1$\\$W_2$\\\vdots\\$W_{\vF}$};
    
    \node [draw,ellipse,below right=1cm and 2cm of serv1] (shared) {Shared Entanglement $\rho_{\mathrm{prev}}$};
    
    \node [right=2cm of user] (receiv) {$W_{K} \in \{0,\ldots,\vM-1\}$};
    
    \path [line] (sent) -- (user);
    
    \path [line] (user.195) --node[pos=0.1,left=2mm] {$Q_1$} (serv1.north);
    \path [line] (user) --node[pos=0.3,left] {$Q_2$} (serv2.north);
    \path [line] (user) --node[pos=0.3,left=2mm] {$Q_{\vN}$} (servn.100);
    
    \path [line] (serv1.80)--node[pos=0.1,right=2mm] {$\cA_1$} (user);
    \path [line] (serv2.83) --node[pos=0.1,right=1mm] {$\cA_2$} (user.287);
    \path [line] (servn.north) --node[pos=0.2,right=4mm] {$\cA_{\vN}$} (user.346);
    
    \path [line] (user.east) -- node{$\rho_{W,Q}$} (receiv);
    
    \path [line,dashed] (shared) --node[pos=0.3,right=6mm] {$\tilde{\cA}_1$} (serv1.south);
    \path [line,dashed] (shared) --node[pos=0.3,right=2mm] {$\tilde{\cA}_2$} (serv2.south);
    \path [line,dashed] (shared) --node[pos=0.33,left=6mm] {$\tilde{\cA}_{\vN}$} (servn.south);
\end{tikzpicture}
}
\caption{Quantum private information retrieval protocol with multiple servers. The composite system of the servers is initialized to an entangled state $\rho_{\mathrm{perv}}$.}
\end{center}
\end{figure}
\end{small}

\subsection{Formal definition of QPIR protocol} \label{sec:problem_statement}

In this paper, we consider the QPIR with multiple servers described as follows.
Let $\vN,\vF,\vM$ be integers greater than $1$.
The participants of the protocol are one user and $\vN$ servers.
The servers do not communicate with each other and each server contains the whole set of uniformly and independently distributed $\vF$ files $W_1,\ldots,W_{\vF} \in \{0,\ldots,\vM-1\}$.
Each server $\mathtt{serv}_t$ possesses a quantum system $\tilde{\cA}_t$ and the $\vN$ servers share an entangled state $\rho_{\mathrm{prev}}\in\mathcal{S}(\bigotimes_{t=1}^{\vN} \tilde{\cA}_t)$.
The user chooses the target file index $K$ to retrieve the $K$-th file $W_K$, 
where the distribution of $K$ is uniform and independent of the files $W_1,\ldots, W_{\vF}$.

To retrieve the $W_K$, 
the user chooses a random variable $R_{\mathrm{user}}$ in a set $\mathcal{R}_{\mathrm{user}}$ and encodes the queries by user encoder $\mathsf{Enc}_{\mathrm{user}}$:
\begin{align*}
 \mathsf{Enc}_{\mathrm{user}}(K,R_{\mathrm{user}}) = (Q_1,\ldots,Q_{\vN}) \in \mathcal{Q}_1\times\cdots\times\mathcal{Q}_{\vN} ,
\end{align*}
where $\mathcal{Q}_t$ is the set of query symbols to the $t$-th server for any $t\in\{1,\ldots,\vN\}$.
The $\vN$ queries $Q_1,\ldots, Q_{\vN}$ are sent to the servers $\mathtt{serv}_1,\ldots,\mathtt{serv}_{\vN}$, respectively.
After receiving the $t$-th query $Q_t$, each server $\mathtt{serv}_t$ 
applies a Completely Positive Trace-Preserving (CPTP) map
$\Lambda_t$ from $\tilde{\cA}_t$ to $\cA_t$ depending on $Q_t,W_1,\ldots,W_{\vF}$
and sends the quantum system $\cA_t$ to the user.
With the server encoder $\mathsf{Enc}_{\mathrm{serv}_t}$, the map $\Lambda_t$ is written as 
\begin{align*}
\Lambda_t = \mathsf{Enc}_{\mathrm{serv}_t}(Q_t,W_1,\ldots,W_{\vF}),
\end{align*}
and the received state of the user is written as
\begin{align}
\rho_{W,Q} \coloneqq \Lambda_1\otimes\cdots\otimes \Lambda_{\vN} (\rho_{\mathrm{prev}}) \in \mathcal{S}\paren*{\bigotimes_{t=1}^{\vN} \cA_t},
\label{def:encodedst}
\end{align}
where $W\coloneqq(W_1,\ldots,W_{\vF})$ and $Q\coloneqq(Q_1,\ldots,Q_{\vN})$.
Next, the user retrieves the file $W_{K}$ by a decoder which is defined depending on $K,Q$ as 
a Positive Operator-Valued Measure (POVM) $\mathsf{Dec}(K,Q)\coloneqq\{\mathsf{Y}_{M}\}_{M=0}^{\vM}$.
The protocol outputs the measurement outcome $M\in\{0,\ldots,\vM\}$ and if $M=\vM$, it is considered as the retrieval failure.

\subsubsection{Protocol}
When the numbers $\vN$ and $\vF$ of servers and files are fixed, 
a QPIR protocol of file size $\vM$ is formulated by the $4$-tuple $\Psi_{\mathrm{QPIR}}^{(\vM)} \coloneqq (\rho_{\mathrm{prev}},\mathsf{Enc}_{\mathrm{user}},{\mathsf{Enc}_{\mathrm{serv}}},\mathsf{Dec})$ of the shared entangled state $\rho_{\mathrm{prev}}$ in the servers, the user encoder $\mathsf{Enc}_{\mathrm{user}}$, the collection of the server encoders $\mathsf{Enc}_{\mathrm{serv}}\coloneqq(\mathsf{Enc}_{\mathrm{serv}_1}, \ldots, \mathsf{Enc}_{\mathrm{serv}_{\vN}})$, and the decoder $\mathsf{Dec}$.

\subsubsection{Security}
For any $t\in\{1,\ldots,\vN\}$,
let $\mathfrak{user}(\Psi_{\mathrm{QPIR}}^{(\vM)})$ and $\mathfrak{serv}_t(\Psi_{\mathrm{QPIR}}^{(\vM)})$ be 
the information of the user and the server $\mathtt{serv}_t$ at the end of the protocol $\Psi_{\mathrm{QPIR}}^{(\vM)}$, respectively.
The security of a QPIR protocol $\Psi_{\mathrm{QPIR}}^{(\vM)}$ is evaluated by the error probability, the server secrecy, and the user secrecy defined as
\begin{align}
P_{\mathrm{err}}(\Psi_{\mathrm{QPIR}}^{(\vM)}) 
&\coloneqq \Pr_{W,K,Q} [M \neq W_K], 
\label{eq:error_prrr}\\
%
S_{\mathrm{serv}}(\Psi_{\mathrm{QPIR}}^{(\vM)}) &\coloneqq I( W_{K^c}; \mathfrak{user}(\Psi_{\mathrm{QPIR}}^{(\vM)}) | K),  \label{ineq:sp}\\
S_{\mathrm{user}}(\Psi_{\mathrm{QPIR}}^{(\vM)}) &\coloneqq \max_{t\in\{1,\ldots,\vN\}} I(K ; \mathfrak{serv}_t(\Psi_{\mathrm{QPIR}}^{(\vM)})), \label{ineq:up}
\end{align}
where $I(\cdot;\cdot|\cdot)$ denotes the conditional mutual information and $W_{K^c} \coloneqq (W_1,\ldots,W_{K-1},W_{K+1},\ldots,W_{\vF})$.
If $S_{\mathrm{serv}}(\Psi_{\mathrm{QPIR}}^{(\vM)})=0$, 
the non-targeted files $W_{K^c}$ are independent of the user information.
Similarly, 
if $S_{\mathrm{user}}(\Psi_{\mathrm{QPIR}}^{(\vM)})=0$, 
the target file index $K$ is independent of any individual server information.

\subsubsection{Costs, rate, and capacity}
Given a QPIR protocol $\Psi_{\mathrm{QPIR}}^{(\vM)}$, we define the upload cost, the download cost, and the QPIR rate by  
\begin{align}
U(\Psi_{\mathrm{QPIR}}^{(\vM)}) &\coloneqq \sum_{t=1}^{\vN} \log |\mathcal{Q}_t|,    \\
D(\Psi_{\mathrm{QPIR}}^{(\vM)}) &\coloneqq \sum_{t=1}^{\vN} \log \dim\cA_t,  \\    
R(\Psi_{\mathrm{QPIR}}^{(\vM)}) &\coloneqq\frac{\log \vM}{D(\Psi_{\mathrm{QPIR}}^{(\vM)})}.
\end{align}
The upload cost, the download cost, and the QPIR rate evaluate {respectively} the size of the whole query set $\mathcal{Q}_1\times \cdots \times \mathcal{Q}_\vN$, the dimension of the downloaded quantum systems $\cA_1\otimes \cdots \otimes \cA_{\vN}$,
and the efficnecy of the protocol. 
When the base of the logarithm is two, the QPIR rate means the number of retrieved bits per one qubit download.

The QPIR capacity is the optimal QPIR rate when the numbers of servers and files are fixed, and we define it with constraints on the security parameters and upload cost.
The {\em asymptotic security-constrained capacity} and the {\em exact security-constrained capacity} are defined 
with $\alpha\in [0,1)$ and  $\beta,\gamma,\theta\in[0,\infty]$ by
\begin{align*}
C_{\mathrm{asymp}}^{\alpha,\beta,\gamma,\theta} 
				& \coloneqq \sup_{\eqref{con1}} 
				\liminf_{\ell\to\infty} R(\Psi_{\mathrm{QPIR}}^{(\vM_\ell)}),\\
C_{\mathrm{exact}}^{\alpha,\beta,\gamma,\theta} 
				& \coloneqq \sup_{\eqref{con2}}
				\liminf_{\ell\to\infty} R(\Psi_{\mathrm{QPIR}}^{(\vM_\ell)}),
\end{align*}
where the supremum is taken for sequences $\{\vM_\ell\}_{\ell=1}^{\infty}$ such that $\lim_{\ell\to\infty} \vM_\ell = \infty$
and for sequences $\{\Psi_{\mathrm{QPIR}}^{(\vM_\ell)}\}_{\ell=1}^{\infty}$ of QPIR protocols
to satisfy either \eqref{con1} or \eqref{con2} given by
\begin{align}    \label{con1} 
\begin{split} 
\!\!\!\limsup_{\ell\to\infty} P_{\mathrm{err}}(\Psi_{\mathrm{QPIR}}^{(\vM_\ell)}) \leq \alpha,  \enskip
& \limsup_{\ell\to\infty} S_{\mathrm{serv}}(\Psi_{\mathrm{QPIR}}^{(\vM_\ell)}) \leq \beta,    \\
\!\!\!\limsup_{\ell\to\infty} S_{\mathrm{user}}(\Psi_{\mathrm{QPIR}}^{(\vM_\ell)}) \leq \gamma, \enskip
& \limsup_{\ell\to\infty} \frac{U(\Psi_{\mathrm{QPIR}}^{(\vM_\ell)})}{D(\Psi_{\mathrm{QPIR}}^{(\vM_\ell)})} \leq  \theta,
\end{split}
\end{align}
and 
\begin{gather}  \label{con2} 
\begin{split}
 P_{\mathrm{err}}(\Psi_{\mathrm{QPIR}}^{(\vM_\ell)}) \leq \alpha, \enskip
& S_{\mathrm{serv}}(\Psi_{\mathrm{QPIR}}^{(\vM_\ell)}) \leq \beta,   \\
 S_{\mathrm{user}}(\Psi_{\mathrm{QPIR}}^{(\vM_\ell)}) \leq \gamma, \enskip
& \limsup_{\ell\to\infty} \frac{U(\Psi_{\mathrm{QPIR}}^{(\vM_\ell)})}{D(\Psi_{\mathrm{QPIR}}^{(\vM_\ell)})} \leq  \theta
.
\end{split}
\end{gather}
It is trivial from the definition that for any $\alpha\in [0,1)$ and $\beta,\gamma,\theta\in[0,\infty]$,
\begin{align}
C_{\mathrm{exact}}^{0,0,0,0}
\leq C_{\mathrm{exact}}^{\alpha,\beta,\gamma,\theta}
\leq C_{\mathrm{asymp}}^{\alpha,\beta,\gamma,\theta} 
\leq C_{\mathrm{asymp}}^{\alpha,\infty,\infty,\infty}.
\label{ineq:capacities}
\end{align}

\subsection{Main Result} \label{sec:main_result}

The main theorem of this paper is given as follows.
\begin{theo} 
When servers can share prior entanglement,
	the capacity of the quantum private information retrieval for $\vF\geq 2$ files and $\vN\geq 2$ servers is
$$ C_{\mathrm{exact}}^{\alpha,\beta,\gamma,\theta}
= C_{\mathrm{asymp}}^{\alpha,\beta,\gamma,\theta} =1,$$
for any $\alpha\in [0,1)$ and $\beta,\gamma,\theta\in[0,\infty]$.
\end{theo}
\begin{proof}
In Sections \ref{sec:achieve} and \ref{sec:converse},
we will prove $C_{\mathrm{exact}}^{0,0,0,0} \geq 1$ and 
$C_{\mathrm{asymp}}^{\alpha,\infty,\infty,\infty} \leq 1$ for any $\alpha\in[0,1)$, respectively.
Then, the inequality \eqref{ineq:capacities} implies the theorem. 
\end{proof}

Note that the capacity does not depend on the number of files $\vF$ and the number of servers $\vN$. 
This contrasts with the classical PIR capacity \cite{SJ16}, which is strictly decreasing for $\vF$ and strictly increasing for $\vN$.
Moreover, the capacity does not depend on the security constraints, i.e., there is no trade-off between the capacity and the constraints $\alpha,\beta,\gamma,\theta$.
Furthermore, the theorem implies that the symmetric QPIR capacity is $1$.

\begin{remark}
In our QPIR model, we assumed that the files $W_1,\ldots, W_{\vN}$ are uniformly random and mutually independent.
		However, the assumption is necessary only for proving the converse bounds.
		Without the assumption, our QPIR protocol has no error and achieves the perfect server and user secrecies.
\end{remark}

\section{Construction of Protocol} \label{sec:achieve}

In this section, we construct a rate-one two-server QPIR protocol with the perfect security and negligible upload cost.
Our protocol is constructed if the file size $\vM$ is the square of an arbitrary integer $\ell$.
Then, by taking $\vM_\ell = \ell^2$, the sequence $\{\Psi_{\mathrm{QPIR}}^{(\vM_\ell)}\}_{\ell=1}^{\infty}$ of our protocols achieves the rate $1$ with 
the perfect security and negligible upload cost, which implies 
\begin{align}
C_{\mathrm{exact}}^{0,0,0,0} \geq 1.  \label{ineq:ac}
\end{align}
In the following, we give preliminaries on quantum operations and states in Section \ref{sec:prelim} and construct the QPIR protocol in Section \ref{sec:protocol}.

\subsection{Preliminaries} \label{sec:prelim}

For an arbitrary integer $\ell\geq 2$, let $\cA$ be an $\ell$-dimensional Hilbert space spanned by an orthonormal basis $\{|0\rangle, \ldots, |\ell-1\rangle\}$.
Define a maximally entangled state $|\Phi\rangle$ on $\cA\otimes \cA$ by
\begin{align*}
|\Phi \rangle \coloneqq  \frac{1}{\sqrt{\ell}} \sum_{i=0}^{\ell-1}| i\rangle \otimes |i \rangle.
\end{align*}
For $a,b\in\mathbb{Z}_\ell$, 
the generalized Pauli operators on $\cA$ are defined as 
\begin{gather*}
\mathsf{X} \coloneqq \sum_{i=0}^{\ell-1} |i+1\rangle \langle i |,\quad
\mathsf{Z} \coloneqq \sum_{i=0}^{\ell-1} \omega^{i} |i\rangle \langle i |,\quad
\end{gather*}
where $\omega = \exp(2\pi\sqrt{-1}/\ell)$, 
and the discrete Weyl operators are defined as
\begin{align*}
\mathsf{W}(a,b) \coloneqq \mathsf{X}^a\mathsf{Z}^b = \sum_{i=0}^{\ell-1} \omega^{ib} |i+a\rangle \langle i | .
\end{align*}
These operators satisfy the relations
\begin{align*}
\mathsf{Z}^b\mathsf{X}^a &= \omega^{ba} \mathsf{X}^a\mathsf{Z}^b,\\
\mathsf{W}(a_1,b_1)\mathsf{W}(a_2,b_2) &= \omega^{b_1a_2} \mathsf{W}(a_1+a_2, b_1+b_2), \\
\mathsf{W}(a,b)^\dagger &= \omega^{ba} \mathsf{W}(-a,-b).
\end{align*}
For any matrix $\mathsf{T} \coloneqq \sum_{i,j=0}^{\ell-1} t_{ij} |i\rangle\langle j|$ on $\cA$,
we define the state $|\mathsf{T}\rangle$ in $\cA\otimes \cA$ by
$$|\mathsf{T}\rangle \coloneqq \sum_{i,j=0}^{\ell-1} t_{ij} |i\rangle \otimes |j\rangle.$$
With this notation, the maximally entangled state is written as $|\Phi\rangle = (1/\sqrt{\ell})|\mathsf{I}\rangle$.
Since $\mathsf{T}^{\top} = \sum_{i,j=0}^{\ell-1} t_{ij} |j\rangle\langle i|$,
it holds $|\mathsf{T}\rangle = (\mathsf{T}\otimes \mathsf{I}) |\mathsf{I}\rangle = (\mathsf{I}\otimes \mathsf{T}^{\top}) |\mathsf{I}\rangle$.
Moreover, for any unitaries $\mathsf{U},\mathsf{V}$ on $\cA$, we have
\begin{align}
(\mathsf{U}\otimes \mathsf{V}) |\mathsf{T}\rangle &= 
|\mathsf{U}\mathsf{T}\mathsf{V}^{\top}\rangle,  \nonumber \\
(\mathsf{U}\otimes \overline{\mathsf{U}}) |\mathsf{I}\rangle &= |\mathsf{U}\mathsf{U}^\dagger\rangle = |\mathsf{I}\rangle. \label{eq:uuI}
\end{align}
With the basis given in the following proposition, we construct the measurement in our QPIR protocol.
\begin{prop}
The set 
$$\mathcal{B} \coloneqq \{ (\mathsf{W}(a,b)\otimes  \mathsf{I} )|\Phi \rangle  \mid a,b\in\mathbb{Z}_\ell \}$$
is an orthonormal basis of $\cA\otimes \cA$.
\begin{proof}
Since $\mathsf{W}(a,b) \otimes \mathsf{I}$ is a unitary matrix for any $a,b\in\mathbb{Z}_\ell$, all elements in $\mathcal{B}$ are unit vectors.
Then, it is sufficient to show that every different two vectors in $\mathcal{B}$ are mutually orthogonal: for any different $(a,b), (c,d)\in\mathbb{Z}_{\ell}^2$,
\begin{align}
( (\mathsf{W}(a,b)\otimes \mathsf{I} )|\Phi \rangle )^{\dagger} (\mathsf{W}(c,d) \otimes \mathsf{I} )|\Phi \rangle = 0 .    \label{eq:ortho}
\end{align}
Since 
$\mathsf{W}(a,b)^{\dagger}\mathsf{W}(c,d) 
= \omega^{b(a-c)}\mathsf{W}(c-a,d-b),$
the left-hand side of \eqref{eq:ortho} is written as 
$$\omega^{b(c-a)}\langle\Phi| (\mathsf{W}(c-a,d-b) \otimes \mathsf{I})|\Phi \rangle.$$
Moreover, for any $x,z\in\mathbb{Z}_{\ell}$, we have 
\begin{align}
\langle \Phi | (\mathsf{W}(x,z)\otimes \mathsf{I}) |\Phi \rangle 
	&= \frac{1}{\ell} \sum_{i=0}^{\ell-1} \langle i | \mathsf{W}(x,z) |i \rangle \\
	&= \frac{1}{\ell} \sum_{i=0}^{\ell-1} \omega^{iz} \langle i |i+x\rangle \\
	&= \delta_{(x,z),(0,0)}
\end{align}
Thus, Eq.~\eqref{eq:ortho} holds for any $(a,b)\neq (c,d)$, which implies the desired statement.
\end{proof}
\end{prop}

\subsection{Rate-one QPIR protocol} \label{sec:protocol}
In this section, we propose a rate-one two-server QPIR protocol with the perfect security and negligible upload cost.
This protocol is constructed from the idea of the classical two-server PIR protocol in \cite[Section 3.1]{CGKS98}.

In this protocol, a user retrieves a file $W_K$ from two servers $\mathtt{serv}_1$ and $\mathtt{serv}_2$.
Each server contains a copy of the files $W_1,\ldots,W_{\vF} \in \{0,\ldots,\ell^2-1 =:\vM_\ell-1\}$ for an arbitrary integer $\ell$. 
By identifying the set $\{0,\ldots,\ell^2-1\}$ with $\mathbb{Z}_\ell^2$, the files $W_1,\ldots,W_{\vF}$ are considered to be elements of $\mathbb{Z}_\ell^2$.
We assume that 
$\mathtt{serv}_1$ and $\mathtt{serv}_2$ possess the $\ell$-dimensional quantum systems $\cA_1$ and $\cA_2$, respectively,
and
the maximally entangled state $|\Phi\rangle$ in $\cA_1\otimes \cA_2$ is shared at the beginning of the protocol.

\subsubsection{Protocol}
The QPIR protocol for retrieving $W_K$ is described as follows.
\begin{enumerate}[Step 1.]
\item Depending on the target file index $K$, the user chooses a subset $R_{\mathrm{user}}$ of $\{1,\ldots,\vF\}$ uniformly.
Let $Q_1\coloneqq R_{\mathrm{user}}$
and 
\begin{align*}
Q_2\coloneqq 
\begin{cases}
Q_1 \setminus \{K\}     &\text{if $K\in Q_1$,}\\
Q_1 \cup \{K\}  &\text{otherwise.}\\
\end{cases}
\end{align*}
\item The user sends the queries $Q_1$ and $Q_2$ to $\mathtt{serv}_1$ and $\mathtt{serv}_2$, respectively.
\item $\mathtt{serv}_1$ calculates $H_1 \coloneqq \sum_{i\in Q_1} W_i \in \mathbb{Z}_\ell^2$
and applies $\mathsf{W}(H_1)$ on the quantum system $\cA_1$.
Similarly, $\mathtt{serv}_2$ calculates $ H_2 \coloneqq \sum_{i\in Q_2} W_i$
and applies $\overline{\mathsf{W}(H_2)}$ to the quantum system $\cA_2$.
The state on $\cA_1\otimes \cA_2$ is $(\mathsf{W}(H_1)\otimes \overline{\mathsf{W}(H_2)}) |\Phi \rangle$.

\item $\mathtt{serv}_1$ and $\mathtt{serv}_2$ send the quantum systems $\cA_1$ and $\cA_2$ 
to the user, respectively.
\item The user performs a POVM
$$\mathsf{Dec}(K,Q)
=
\{ \mathsf{Y}_{(a,b)} \mid a,b\in\mathbb{Z}_\ell \}$$ on the received state $\rho_{W,Q}$,
where each POVM element $\mathsf{Y}_{(a,b)}$ for the outcome $(a,b)$ is defined by
\begin{align*}
\mathsf{Y}_{(a,b)} \coloneqq (\mathsf{W}(a,b) \otimes \mathsf{I})|\Phi \rangle 
\langle \Phi | (\mathsf{W}(a,b)^{\dagger} \otimes \mathsf{I})  
\end{align*} if $K\in Q_1$, and
\begin{align*}
\mathsf{Y}_{(a,b)} \coloneqq (\mathsf{W}(-a,-b) \otimes \mathsf{I})|\Phi \rangle 
\langle \Phi | (\mathsf{W}(-a,-b)^{\dagger} \otimes \mathsf{I} )
\end{align*}
otherwise.
The user obtains the measurement outcome $(a,b)$ as the retrieval result.
\end{enumerate}

\subsubsection{Security}
We analyze the security of the protocol, i.e., the error probability, the server secrecy, and the user secrecy.

The protocol has no error as follows.
Note that $H_1 = H_2 + W_K$ if $K\in Q_1$, and $H_1=H_2-W_K$ otherwise.
After Step 3, the state on $\cA_1\otimes\cA_2$ is
\begin{align}
 \lefteqn{\mathsf{W}(H_1)\otimes \overline{\mathsf{W}(H_2)} |\Phi \rangle }  \nonumber \\
 &= 
 \frac{\omega^{\mp b_{W_K}a_{H_2}}}{\sqrt{\ell}} (\mathsf{W}(\pm W_K) \otimes \mathsf{I}) ( \mathsf{W}(H_2)\otimes \overline{\mathsf{W}(H_2)}) | \mathsf{I} \rangle \label{123} \\
 &= \frac{\omega^{\mp b_{W_K}a_{H_2}}}{\sqrt{\ell}} (\mathsf{W}(\pm W_K) \otimes \mathsf{I}) | \mathsf{I} \rangle \label{234}\\
 &= \omega^{\mp b_{W_K}a_{H_2}} (\mathsf{W}(\pm W_K) \otimes \mathsf{I})| \Phi \rangle,     \nonumber
\end{align}
where $H_2 =(a_{H_2}, b_{H_2})$ and  $W_K=(a_{W_K}, b_{W_K})\in\mathbb{Z}_\ell^2$.
The equality \eqref{123} is derived from $\mathsf{W}(H_1) = \mathsf{W}(\pm W_K+H_2) = \omega^{\mp b_{W_K}a_{H_2}} \mathsf{W}(\pm W_K)\mathsf{W}(H_2)$
and the equality \eqref{234} is from \eqref{eq:uuI}.
Therefore, in Step 5, the measurement outcome is $W_K\in\mathbb{Z}_\ell^2$ with probability $1$.

The perfect server secrecy is obtained because 
the received state $(\mathsf{W}(\pm W_K)\otimes \mathsf{I})|\Phi\rangle$ of the user is independent of the files $W_1,\ldots,W_{K-1}, W_{K+1},\ldots, W_{\vF}$.

The perfect user secrecy follows from that of the protocol \cite[Section 3.1]{CGKS98}.
Note that even if the collection of $Q_1$ and $Q_2$ depends on $K$, 
	each of $Q_1$ and $Q_2$ is individually independent of the index $K$. 
	Thus, the perfect user secrecy is obtained.


\subsubsection{Upload cost, download cost, and rate} 
The upload cost is $U(\Psi_{\mathrm{QPIR}}^{(\vM_\ell)})= {2\vF}\log 2$ since two subsets $Q_1$ and $Q_2$ of ${\{1,\ldots, \vF\}}$ are uploaded and each subset of ${\{1,\ldots, \vF\}}$ is expressed by $\vF$ bits.
The download cost is $D(\Psi_{\mathrm{QPIR}}^{(\vM_\ell)})=\log \dim\cA_1\otimes \cA_2 = \log \ell^2 = \log \vM_\ell$.
Therefore, the rate is 
\begin{align*}
R(\Psi_{\mathrm{QPIR}}^{(\vM_\ell)}) = \frac{\log \vM_\ell}{D(\Psi_{\mathrm{QPIR}}^{(\vM_\ell)})} = 1,
\end{align*}
and 
$U(\Psi_{\mathrm{QPIR}}^{(\vM_\ell)})/D(\Psi_{\mathrm{QPIR}}^{(\vM_\ell)})$ goes to zero as $\vM_\ell\to\infty$.



\subsection{Security against malicious operations}

In the previous subsection, we showed that the protocol in Section \ref{sec:protocol} 
has the perfect security when the user and the servers follow the protocol.
In this subsection, we prove that the protocol in Section \ref{sec:protocol} also guarantees the server and user secrecies even if the servers or the user apply malicious operations.
Namely, we consider two malicious models: the malicious server model and the malicious user model. 

The malicious server model considers the case that the servers apply malicious operations to obtain the target file index $K$
but the user follows the protocol, i.e., the query generation and the recovery by the user are the same as the protocol in Section \ref{sec:protocol}.
Our protocol is secure against this model since
	each of $Q_1$ and $Q_2$ is individually independent of the index $K$
	and the servers obtain no more information from the user except for $Q_1$ and $Q_2$.
Therefore the servers cannot obtain any information of $K$ by malicious operations.

The second security model is the malicious user model, where the user sends malicious queries to the servers to obtain the non-targeted file information in addition to the target file $W_K$.
That is, the user sends malicious queries $Q = (Q_1, Q_2)$ to retrieve both of the file $W_K$ and some information of $W_{K^{c}} = (W_1,\ldots,W_{K-1}, W_{K+1},\ldots, W_{\vF})$.
Similarly to the malicious server model, we assume that the servers do not deviate from the protocol.
Our protocol is also secure against this model
since the user downloads the $\vM_{\ell}$-dimensional quantum system and the user is assumed to obtain $W_K \in \{0,\ldots, \vM_{\ell}-1\}$.
That is, the user cannot obtain more information than $W_K$.
This security is precisely proved by the following relation:
\begin{align}
{I(\cA;W_{K^{c}}|W_K,K,Q)_{\rho_{W,Q}} = 0,} \label{noinf}
\end{align}
where $\cA = \cA_1\otimes \cA_2$
and $I(\cdot;\cdot |\cdot)_{\rho}$ is the quantum conditional mutual information defined in Appendix \ref{append:entropy}.

\begin{proof}[Proof of Eq.~\eqref{noinf}]
Since the user obtains the file $W_K$, we have
\begin{align}
H(W_K|\cA,K,Q)_{\rho_{W,Q}} &= 0,   \label{eq:asdfff}
\end{align}
where $H(\cdot|\cdot)_{\rho}$ is the quantum conditional entropy defined in Appendix \ref{append:entropy}.
Eq. \eqref{eq:asdfff} is equivalent to 
\begin{align}
H(\cA,W_K|K,Q)_{\rho_{W,Q}} &= H(\cA|K,Q)_{\rho_{W,Q}}.   \label{eq:samesame}
\end{align}
The relation \eqref{eq:samesame} implies the following relations:
\begin{align}
0&\leq H(\cA|W_K,K,Q)_{\rho_{W,Q}}  \\
                &= H(\cA,W_K|K,Q)_{\rho_{W,Q}} - H(W_K|K,Q) \\
                &= H(\cA|K,Q)_{\rho_{W,Q}} -\log \vM_\ell \leq 0. \label{eq:lastt}
\end{align}
The equality in \eqref{eq:lastt} follows from the condition \eqref{eq:samesame}, the independence between $W_K$ and $(K,Q)$, and the uniform distribution of $W_K$.
The last inequality in \eqref{eq:lastt} follows from $\dim \cA = \log \vM_\ell$.
Therefore, we have
\begin{align}
H(\cA|W_K,K,Q)_{\rho_{W,Q}} = 0
\end{align}
which implies \eqref{noinf}.
\end{proof}


%
%
%
%

\section{Converse} \label{sec:converse}

In this section, we prove the converse bound 
\begin{align}
C_{\mathrm{asymp}}^{\alpha,\infty,\infty,\infty} \leq 1 \label{ineq:converse_im}
\end{align}
for any $\alpha\in[0,1)$.
By replacing the notation of $\rho_{W,Q}$ defined in \eqref{def:encodedst},
let $\rho_{w,z}$ be the quantum state on the composite system $\bigotimes_{t=1}^{\vN} \cA_t$,
where $w$ is the file to be retrieved
and
$z \coloneqq (w^c,q)$ for the collection $w^c$ of other $\vM-1$ files and
the collection $q$ of queries.

Applying \cite[(4.66)]{Hay17} to the choice $\sigma_{z}=(1/\vM) \sum_{w=0}^{\vM-1} \rho_{w,z}$,
for any $s \in (0,1)$, we have
\begin{align}
 (1-P_{\mathrm{err},z}(\Psi_{\mathrm{QPIR}}^{(\vM)}) )^{1+s} \vM^{s}
 \le  \frac{1}{\vM} \sum_{w=0}^{\vM-1} \Tr \rho_{w,z}^{1+s} \sigma_{z}^{-s},
 \label{eq:feff}
\end{align}
where $P_{\mathrm{err},z}(\Psi_{\mathrm{QPIR}}^{(\vM)})$ is the error probability when $z$ is fixed.
For the completeness of the proof, we give the derivation of \eqref{eq:feff} in Appendix~\ref{append:avvvv}.
Furthermore, we can bound the RHS of \eqref{eq:feff} as 
\begin{align}
 &\frac{1}{\vM} \sum_{w=0}^{\vM-1} \Tr \rho_{w,z}^{1+s} \sigma_{z}^{-s}
 \le  \frac{1}{\vM} \sum_{w=0}^{\vM-1} \Tr \rho_{w,z} \sigma_{z}^{-s}
 =  \Tr \sigma_{z}^{1-s}\nonumber  \\
 &\le \max_{\sigma}\Tr \sigma^{1-s}
 = \max_{p} \sum_{i=1}^d p_i ^{1-s}
 \stackrel{\mathclap{(a)}}{=} 
 \paren*{\prod_{t=1}^{\vN} \dim \cA_t}^s
 \label{eq:jojh}
\end{align}
for $d= \prod_{t=1}^{\vN} \dim \cA_t$. 
Here, since $x \mapsto x^{1-s}$ is concave,
the maximum $\max_{p} \sum_{i=1}^d p_i^{1-s}$
is realized by the uniform distribution, which shows the equation $(a)$.
Combining \eqref{eq:feff} and \eqref{eq:jojh}, the error probability is upper bounded as 
\begin{align}
1-P_{\mathrm{err}}(\Psi_{\mathrm{QPIR}}^{(\vM)})  
	&= 1-\mathbb{E}_Z P_{\mathrm{err},Z}(\Psi_{\mathrm{QPIR}}^{(\vM)})  \\
	&\le \paren*{\frac{{\prod_{t=1}^{\vN} \dim \cA_t}}{\vM}}^{\frac{s}{1+s}}.
	\label{eq:fffefw}
\end{align}

For any sequence of QPIR protocols 
$\{\Psi_{\mathrm{QPIR}}^{(\vM_\ell)}\}_{\ell=1}^{\infty}$,
if $\Psi_{\mathrm{QPIR}}^{(\vM_\ell)}$
satisfies
\begin{align}
R(\Psi_{\mathrm{QPIR}}^{(\vM_\ell)}) 
=
\frac{\log \vM_\ell}{\log \prod_{t=1}^{\vN} \dim\cA_t}
\geq 1
\end{align}
for any sufficiently large $\ell$,
we have $$\frac{\prod_{t=1}^{\vN} \dim \cA_t}{\vM_\ell} \to 0.$$
Hence, by \eqref{eq:fffefw}, $1-P_{\mathrm{err}}(\Psi_{\mathrm{QPIR}}^{(\vM_\ell)})$ approaches zero,
which implies \eqref{ineq:converse_im}.

\section{Capacity of Multi-Round QPIR} \label{sec:multi}

In this section, we prove that the multi-round QPIR capacity is $1$.
First, as a generalization of the protocol description in Section \ref{sec:problem_statement}, we formally define the multi-round QPIR.
Then, we propose the capacity theorem and the proof of the weak converse bound.
Since the protocol in Section \ref{sec:protocol} has the QPIR rate $1$, this protocol also achieves the multi-round QPIR capacity.
The result in this section includes the result in the previous sections as the one-round QPIR.

\begin{figure*}[t]
\begin{center}
        \resizebox {\linewidth} {!} {
\begin{tikzpicture}[scale=0.5, node distance = 3.3cm, every text node part/.style={align=center}, auto]

	
	\node [block, minimum height=6em] (enc1) {\small $\mathsf{Enc}_{\mathtt{serv}}^{(1),Q^{[1]},W}$};
	
	\node [below left=1em and 7em of enc1.145] (start) {$\rho_{\mathrm{prev}}$};
	\node [above left=1em and 7em of enc1.145] (W)  {$W$};
	\node [below left=10em and 7em of enc1.145] (K)  {$K$};
	
	\node [below left=3.2em and 0em of start] (lefts) {};
	\node [right=48em of lefts] (rights) {};
	\path [line,-,dashed] (lefts) -- node[below,pos=0.99]{User} node[above,pos=0.99]{Servers} (rights);
	
	\node [block, minimum height=6em, right=13.5em of enc1] (enc2) {\small $\mathsf{Enc}_{\mathtt{serv}}^{(2),Q^{[2]},W}$};
	
	\node [block, minimum height=6em, below right=3.5em and 4em of enc1] (dec1) {$\mathsf{Dec}^{(1),K,Q^{[1]}}$};
	\node [block, minimum height=6em, right=13.5em of dec1] (dec2) {$\mathsf{Dec}^{(2),K,Q^{[2]}}$};

	\path [line,stealth-] (enc1.145) node[left=2em,above]{$\mathcal{B},W$}--++(0:-9em)|-  (start.east);
	\path [line,-] (enc1.145) --++(0:-9em)|- (W.east) ;

	\path [line] (K.east) --++(0:5em) --node[above,above left=1.5em and -1.5em]{$Q^{1}$} (enc1.215);
	\path [line] (K.east) --++(0:5em)|- node[right=6.5em, above]{$\mathcal{C},K,Q^{[1]}$} (dec1.215);
	
	\path [line] (enc1.35) --node[above]{$\mathcal{B},W,Q^{[1]}$} (enc2.145);
	\path [line] (enc1.325) -- node[below,pos=0.7,left=0.1em]{$\mathcal{A}^{1}$} (dec1.145);
	\path [line] (enc2.325) -- node[below,pos=0.7,left=0.1em]{$\mathcal{A}^{2}$} (dec2.145);


	\path [line] (dec1.35) --node[above,above left=0.9em and -1.4em]{$Q^{2}$} (enc2.215);
	\path [line] (dec1.325) -- node[above]{$\mathcal{C},K,Q^{[2]}$} (dec2.215);

	\path [line] (enc2.35) --node[above]{$\mathcal{B},W,Q^{[2]}$} ++(0:10em);
	\path [line] (dec2.east) -- ++(0:5em) node[right] {$W_K$};

\end{tikzpicture}
}
\caption{The information flow in $2$-round QPIR protocol.
The servers have all files $W =(W_1,\ldots, W_{\vF})$ and the user retrieves the $K$-th file $W_K$.
}
\label{fig:flow}
\end{center}
\end{figure*}
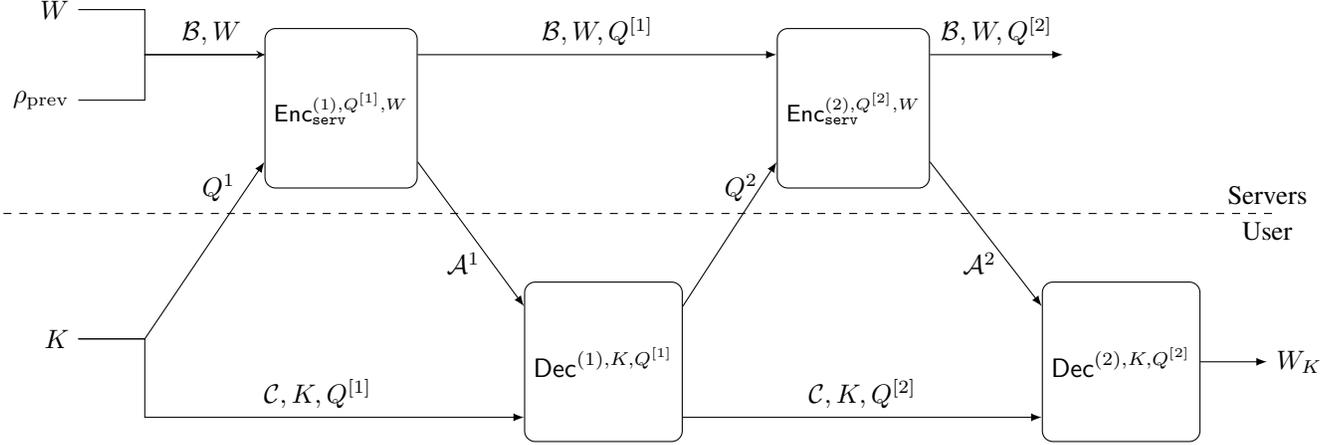

\subsection{Formal Definition of Multi-Round QPIR Protocol} \label{sec:mutli_def}
For any positive integer $\vR$,
we give the formal description of the $\vR$-round QPIR protocol $\Psi^{(\vM,\vR)}_{\mathrm{QPIR}}$.
The information flow of the quantum systems is depicted in Fig.~\ref{fig:flow}.
When $\vR=1$, the protocol description is equivalent to the protocol defined in Section \ref{sec:problem_statement}.

Let $\vN,\vF,\vM$ be integers greater than $1$.
Each of the servers $\mathtt{serv}_1,\ldots,\mathtt{serv}_{\vN}$ contains the whole copy of the uniformly and independently distributed $\vF$ files $W=(W_1,\ldots, W_{\vF}) \in\{0,\ldots, \vM-1\}^{\vF}$.
The $t$-th server $\mathtt{serv}_t$ possesses a quantum system ${\cB}_t$ as local quantum register
and the $\vN$ servers share an entangled state $\rho_{\mathrm{prev}}$ on the quantum system $\cB \coloneqq {\cB}_1\otimes \cdots \otimes {\cB}_{\vN}$.

The user 
chooses the target file index $K\in\{1,\ldots,\vF\}$ uniformly and independently of the files $W_1,\ldots, W_{\vF}$.
The user prepares the query $Q^{1}= (Q_1^{1} , Q_2^{1}, \ldots, Q_{\vN}^{1})$ depending on $K$.
The user has a local quantum register $\mathcal{C}$ where the state is initialized depending on $K$ and $Q^{1}$.

For $i\in\{1,\ldots, \vR\}$, the $i$-th round is described as follows.
Let $Q^{i}_t$ be the query to $\mathtt{serv}_t$ at round $i$, and we denote $Q^i\coloneqq(Q^i_1,\ldots, Q^i_{\vN})$ and $Q_t^{[i]} \coloneqq (Q_t^{1},\ldots,Q_t^{i} )$.
The query $Q^i$ for round $i$ is determined at round $i-1$.
The user sends ${Q}_t^{i}$ to the $t$-th server $\mathtt{serv}_t$.
Depending on ${Q}_t^{[i]}$ and $W$,
each server $\mathtt{serv}_t$ applies a CPTP map $\mathsf{Enc}_{\mathtt{serv}_t}^{(i), Q_t^{[i]},W}$ from $\cB_t$ to $\cA_t^{i}\otimes  \mathcal{B}_t$.
That is, when the collection of the encoders is written as 
\begin{align*}
\mathsf{Enc}_{\mathtt{serv}}^{(i), Q^{[i]},W} \coloneqq \bigotimes_{t=1}^{\vN} \mathsf{Enc}_{\mathtt{serv}_t}^{(i), Q_t^{[i]},W},
\end{align*}
the state $\rho_W^{\mathcal{B}}$ on $\mathcal{B}$
is encoded as
\begin{align*}
\rho_W^{\cA^{i}\mathcal{B}} \coloneqq \mathsf{Enc}_{\mathtt{serv}}^{(i), Q^{[i]},W} (\rho_W^{\mathcal{B}})  , 
\end{align*}
where $\cA^{i} \coloneqq \cA_1^{i}\otimes \cdots\otimes \cA_{\vN}^{i}$.
Each server transmits the system $\mathcal{A}_t^{i}$ to the user and 
the received state of the user is the reduced state
\begin{align}
\rho_W^{\cA^{i}} \coloneqq \Tr_{\mathcal{B}} \rho_W^{\cA^{i}\mathcal{B}}.
\label{def:encodedst2}
\end{align}
If $ i < \vR$, the user applies 
a quantum instrument $\mathsf{Dec}^{(i),K,Q^{[i]}} = \{Y^{i}_{Q^{i+1}}\}_{Q^{i+1}\in \mathcal{Q}^{i+1}}$ from $
\mathcal{A}^{i} \otimes \mathcal{C}$ to ${\mathcal{C}}$ depending on $K$ and $Q^{[i]} \coloneqq (Q^1,\ldots, Q^i)$,
where $\mathcal{Q}_1^{i+1}\times \cdots \times \mathcal{Q}_{\vN}^{i+1}$ is the set of queries at round $i+1$
	and $Q^{i+1}$ is the measurement outcome.
Then round $i$ ends and round $i+1$ starts.
If $i=\vR$, i.e., at the final round, 
	the user applies a POVM $\mathsf{Dec}^{(\vR),K,Q^{[\vR]}} = \{\mathsf{Y}_{M}\}_{M=0}^{\vM}$ on $\cA^{\vR}\otimes \mathcal{C}$
	depending on $K$ and $Q^{[\vR]}$
	and
	outputs the measurement outcome $M\in\{0,\ldots,\vM\}$.
	If $M=\vM$, it is considered as the retrieval failure.

Similarly to Section \ref{sec:problem_statement}, 
the security of the protocol is evaluated by the the error probability, the server secrecy, and the user secrecy defined by
\begin{align*}
P_{\mathrm{err}}(\Psi_{\mathrm{QPIR}}^{(\vM,\vR)}) &\coloneqq 
\Pr_{W,K,Q^1}[ M \neq W_K],
\\
S_{\mathrm{serv}}(\Psi_{\mathrm{QPIR}}^{(\vM,\vR)}) &\coloneqq I( W_{K^c}; \mathfrak{user}(\Psi_{\mathrm{QPIR}}^{(\vM,\vR)}) | K),\\
S_{\mathrm{user}} (\Psi_{\mathrm{QPIR}}^{(\vM,\vR)}) &\coloneqq \max_{t\in\{1,\ldots,\vN\}}I(K ; \mathfrak{serv}_t(\Psi_{\mathrm{QPIR}}^{(\vM,\vR)})).
\end{align*}
Given the QPIR protocol $\Psi_{\mathrm{QPIR}}^{(\vM,\vR)}$, we define the upload cost, the download cost, and the QPIR rate by  
\begin{align}
U(\Psi_{\mathrm{QPIR}}^{(\vM,\vR)})&\coloneqq \sum_{i=1}^{\vR} \log |\mathcal{Q}^{i}|, \\ 
D(\Psi_{\mathrm{QPIR}}^{(\vM,\vR)})&\coloneqq \sum_{i=1}^{\vR} \log \dim\cA^{i},  \\
R(\Psi_{\mathrm{QPIR}}^{(\vM,\vR)})&\coloneqq\frac{\log \vM}{D(\Psi_{\mathrm{QPIR}}^{(\vM,\vR)})}.
\end{align}

Now, we define the $\vR$-round QPIR capacities with four parameters as follows.
For an error constraint $\alpha\in [0,1)$, 
server secrecy constraint $\beta\in[0,\infty]$,
user secrecy constraint $\gamma\in[0,\infty]$,
and upload constraint $\theta\in[0,\infty]$, 
the {\em asymptotic security-constrained $\vR$-round capacity} and the {\em exact security-constrained $\vR$-round capacity} are defined as
\begin{align*}
C_{\mathrm{asymp}}^{\alpha,\beta,\gamma,\theta} 
				& \coloneqq \sup_{\eqref{con3}} 
				\liminf_{\ell\to\infty} R(\Psi_{\mathrm{QPIR}}^{(\vM_\ell,\vR)}),\\
C_{\mathrm{exact}}^{\alpha,\beta,\gamma,\theta} 
				& \coloneqq \sup_{\eqref{con4}}
				\liminf_{\ell\to\infty} R(\Psi_{\mathrm{QPIR}}^{(\vM_\ell,\vR)}),
\end{align*}
where the supremum is taken for sequences $\{\vM_\ell\}_{\ell=1}^{\infty}$ such that $\lim_{\ell\to\infty} \vM_\ell = \infty$
and for sequences $\{\Psi_{\mathrm{QPIR}}^{(\vM_\ell,\vR)}\}_{\ell=1}^{\infty}$ of $\vR$-round QPIR protocols
to satisfy either \eqref{con1} or \eqref{con2} given by
\begin{align}    \label{con3} 
\begin{split} 
\!\!\!\!\limsup_{\ell\to\infty} P_{\mathrm{err}}(\Psi_{\mathrm{QPIR}}^{(\vM_\ell,\vR)}) \leq \alpha,  \enskip
& \limsup_{\ell\to\infty} S_{\mathrm{serv}}(\Psi_{\mathrm{QPIR}}^{(\vM_\ell,\vR)}) \leq \beta,    \\
\!\!\!\!\limsup_{\ell\to\infty} S_{\mathrm{user}}(\Psi_{\mathrm{QPIR}}^{(\vM_\ell,\vR)}) \leq \gamma, \enskip
& \limsup_{\ell\to\infty} \frac{U(\Psi_{\mathrm{QPIR}}^{(\vM_\ell,\vR)})}{D(\Psi_{\mathrm{QPIR}}^{(\vM_\ell)})} \leq  \theta,
\end{split}
\end{align}
and 
\begin{gather}  \label{con4} 
\begin{split}
 P_{\mathrm{err}}(\Psi_{\mathrm{QPIR}}^{(\vM_\ell,\vR)}) \leq \alpha, \enskip
& S_{\mathrm{serv}}(\Psi_{\mathrm{QPIR}}^{(\vM_\ell,\vR)}) \leq \beta,   \\
 S_{\mathrm{user}}(\Psi_{\mathrm{QPIR}}^{(\vM_\ell,\vR)}) \leq \gamma, \enskip
& \limsup_{\ell\to\infty} \frac{U(\Psi_{\mathrm{QPIR}}^{(\vM_\ell,\vR)})}{D(\Psi_{\mathrm{QPIR}}^{(\vM_\ell,\vR)})} \leq  \theta
.
\end{split}
\end{gather}

The multi-round QPIR capacity is derived as follows.
\begin{theo}
[Multi-round QPIR capacity]
Let $\vR$ be any positive integer.
When servers can share prior entanglement,
	the $\vR$-round QPIR capacity for $\vF\geq 2$ files and $\vN\geq 2$ servers is
\begin{align}
C_{\mathrm{exact}}^{0,\beta,\gamma,\theta,\vR}
= C_{\mathrm{asymp}}^{0,\beta,\gamma,\theta,\vR} =1     \label{1111}
\end{align}
for any $\beta,\gamma,\theta\in[0,\infty]$.
\end{theo}

\begin{proof}

Eq. \eqref{1111} is proved by the following inequalities:
\begin{align*}
 1\leq C_{\mathrm{exact}}^{0,0,0,0,\vR} 
 \leq C_{\mathrm{exact}}^{0,\beta,\gamma,\theta,\vR}
 \leq C_{\mathrm{asymp}}^{0,\beta,\gamma,\theta,\vR}
 \leq C_{\mathrm{asymp}}^{0,\infty,\infty,\infty,\vR} \leq  1.
\end{align*}
The first inequality holds by applying the rate-one QPIR protocol in Section \ref{sec:protocol} repetitively $\vR$ times.
The second, third, and fourth inequalities follow from the definition of the capacities.
The last inequality is proved in Section \ref{subsec:weak}.
Therefore, we obtain the theorem.
\end{proof}

\subsection{Weak converse bound of multi-round QPIR capacity} \label{subsec:weak}

We prove the converse bound
\begin{align}
C_{\mathrm{asymp}}^{0,\infty,\infty,\infty,\vR} \leq  1. \label{weakconverse}
\end{align}
Our proof comes from the fact that 
	the multi-round QPIR protocol can be considered as a case of the Classical-Quantum (CQ) channel coding with classical feedback \cite{DQSW19}.
	In the CQ channel coding with classical feedback,
		the sender encodes a classical message $W$ as a quantum state 
		and sends the state over a fixed channel $\mathcal{N}$.
	The receiver performs a decoding measurement on the received state and returns the measurement outcome to the sender. 
	The sender and the receiver iterate this process $\vR$ times while using the previous measurement outcomes for encoding and decoding. 
	At the end of the protocol, the receiver receives the classical message $W$.
	The paper \cite{DQSW19} proved that the capacity of this problem when the sender and the receiver have their local quantum registers, respectively.
	More specifically, the paper \cite{DQSW19} also considered the energy constraint $E$ that 
	for a given Hamiltonian $\mathsf{H}$ on the input system of $\mathcal{N}$, 
	the input states $\rho_1,\ldots, \rho_{\vR}$ to $\mathcal{N}$ should satisfy $\sum_{i=1}^{\vR} \Tr \rho_i \mathsf{H} \le E$.
	The CQ channel capacity is characterized by the following proposition.
	\begin{prop}[{\cite[Theorem 4]{DQSW19}}]
	Let $\mathcal{N}$ be a quantum channel,
		$\vR$ be the number of communication rounds,
		$\vM$ be the size of the message set,
		$\mathsf{H}$ be the Hamiltonian on the input system of $\mathcal{N}$,
		$E$ be the energy constraint,
		and $\varepsilon$ be the error probability.
	Suppose the sender and the receiver have local quantum registers, respectively.
	Then, for the CQ channel coding with classical feedback and energy constraint,
		we have the following inequality:
		\begin{align}
		(1-\varepsilon) \log \vM &\le \sup_{\rho: \Tr \rho \mathsf{H} \le E} \vR H(\mathcal{N}(\rho)) + h_2(\varepsilon) .
		\end{align}
	\end{prop}

	The multi-round QPIR protocol can be considered as a case of this problem where the channel $\mathcal{N}$ is the identity channel and there is no energy constraint.
	To see this fact, we 
		consider the the collection of the servers as the sender and the user as the receiver of the CQ channel coding,
		and focus on the communication of a classical message from the collection of the servers to the user.
	The servers sends to the user the systems $\cA^{i}$ over the identity channel and the user sends queries $Q^{i}$ to the servers as the measurement outcome on $\cA^{i}$.
	The servers and the user have $\cB$ and $\mathcal{C}$ as local quantum registers, respectively.
	At the end of the protocol, the user obtains the classical target file $W_K$.
	Therefore, we can consider the multi-round QPIR protocol as a CQ channel coding with classical feedback.
	
	By the similar proof of \cite[Theorem 4]{DQSW19}, we have the following proposition.
\begin{prop} \label{prop:cqfeed}
Consider the CQ channel coding of a classical message $W\in\{0,\ldots,\vM-1\}$ from the sender to the receiver by sending quantum systems $\cA^{1},\ldots,\cA^{\vR}$ sequentially over the identity channel and assisted by classical feedback.
We assume that the sender and the receiver have local quantum registers, respectively.
Let $\rho_W^{\cA^{i}}$ be the state on $\cA^{i}$.
For the uniformly chosen message $W$ and the decoding output $M$,
we define the error probability $\varepsilon \coloneqq \Pr[M \neq W]$.
Then we have the following inequality 
\begin{align}
(1-\varepsilon) \log \vM &\le \sum_{i=1}^{\vR} H(\rho_W^{\cA^{i}}) + h_2(\varepsilon) 
\label{eq:multi-converse-label0}
\\
    &\le \sum_{i=1}^{\vR} \log \dim \cA^{i} + h_2(\varepsilon),
\label{eq:multi-converse-label}
\end{align}
where $h_2(\cdot)$ is the binary entropy function. 
\end{prop}

For the completeness of our paper, we give a proof of Proposition~\ref{prop:cqfeed} in Appendix~\ref{append:cqfeed}.

\begin{remark}
Proposition \ref{prop:cqfeed} is slightly different from \cite[Theorem 4]{DQSW19}.
First, whereas \cite[Theorem 4]{DQSW19} considers an energy constraint on the quantum channel,  
Proposition \ref{prop:cqfeed} assumes no energy constraint.
Second, whereas \cite[Theorem 4]{DQSW19} considers the repetitive uses of a fixed quantum channel $\mathcal{N}$,
Proposition \ref{prop:cqfeed} considers each use of the identity quantum channels over $\cA^{1}, \ldots, \cA^{\vR}$.
Even with these differences, 
	we can apply the same proof steps of \cite[Theorem 4]{DQSW19}
	and the first inequality of \cite[Eq.~(35)]{DQSW19} is the inequality \eqref{eq:multi-converse-label0}.
\end{remark}

Now we prove the weak converse bound. 
We choose an arbitrary sequence $\{\Psi_{\mathrm{QPIR}}^{(\vM_\ell,\vR)}\}_{\ell=1}^{\infty}$ of $\vR$-round QPIR protocols
	to satisfy $\varepsilon_\ell \coloneqq P_{\mathrm{err}}(\Psi_{\mathrm{QPIR}}^{(\vM_\ell,\vR)}) \to 0$ as $\ell\to\infty$.
	Considering the collection of the $\vN$ servers as the sender and the user as the receiver of Proposition \ref{prop:cqfeed}, 
	we can apply Proposition \ref{prop:cqfeed} to the $\vR$-round QPIR protocol $\Psi_{\mathrm{QPIR}}^{(\vM_\ell,\vR)}$ 
	with the classical message $W_K \in \{0,\ldots,\vM_\ell-1\}$, the transmitted quantum systems $\cA^{1},\ldots,\cA^{\vR}$, and the classical feedbacks $Q^{1}, \ldots, Q^{\vR}$.
	In this case, $\varepsilon$ and $\vM$ of Proposition \ref{prop:cqfeed} is substituted by $\varepsilon_{\ell}$ and $\vM_{\ell}$, i.e., Eq.~\eqref{eq:multi-converse-label} is written as
\begin{align}
(1-\varepsilon_\ell) \log \vM_\ell &\le \sum_{i=1}^{\vR} \log \dim \cA^{i} + h_2(\varepsilon_\ell).
\end{align}
Therefore, we have
\begin{align}
\lim_{\ell\to\infty} R(\Psi_{\mathrm{QPIR}}^{(\vM_\ell,\vR)})= \lim_{\ell\to\infty} \frac{\log \vM_\ell}{\sum_{i=1}^{\vR} \log \dim \cA^{i}} \le 1,
\end{align}
which implies \eqref{weakconverse}.


\section{Conclusion} \label{sec:conclusion}

We have studied the capacity of QPIR with multiple servers when the servers share prior entanglement.
Considering the user secrecy and the server secrecy,
	we defined two kinds of QPIR capacities: asymptotic and exact security-constrained capacities with upload constraint. 
We proved that both QPIR capacities are $1$ for any security constraints and any upload constraint.
We have constructed a capacity-achieving rate-one protocol only with two servers when the file size is the square of an arbitrary integer.
The converse has been proved by focusing on the download step of QPIR protocols.
{Furthermore, we have proved that the capacity of multi-round QPIR is also $1$ by the weak converse bound.}

It is an interesting open question whether
the QPIR without shared entanglement also has an advantage over the classical PIR counterparts.
This paper has considered the QPIR under the assumption of the prior entanglement and 
	the quantum capacity of this case is strictly greater than the classical PIR capacity.
The QPIR capacity without shared entanglement lies between the QPIR capacity of this paper and the classical PIR capacity.
Therefore, it should be studied whether the quantum PIR capacity is strictly higher than the classical PIR capacity even without shared entanglement.

In this paper, we have assumed that the maximally entangled state can be shared by several servers.
That is, we have made no restriction for shared entanglement.
This setting is similar to 
the original studies \cite{bennett1999entanglement,bennett2002entanglement} for the entangled assisted classical capacity for a noisy quantum channel
because they have no restriction for shared entanglement.
The recent paper \cite{ISIT2020} derived the entanglement-assisted classical capacity for a noisy quantum channel when shared entanglement is limited.
For the extension, the study \cite{ISIT2020} 
invented several new methods, which are essential for this restriction.
Therefore, it is remained as a future problem to extend our result to 
the case when the shared entangled state is restricted.

As another problem, we can consider the QPIR capacity when
the channel from servers to user are noisy quantum channels.
It is natural to apply quantum error corrections to each noisy quantum channels
and apply our QPIR protocol over the virtually implemented noiseless channels by error correction.
In this case, the transmission rate is given by the quantum capacity of the noisy quantum channel.
For the converse part, we can easily extend the discussion of Section \ref{sec:converse}.
In this extension, the obtained upper bound of the transmission rate is 
the classical capacity of the noisy quantum channel.
Hence, this simple method does not yield the QPIR capacity with noisy quantum channels.
Therefore, it is another challenging problem to calculate the QPIR capacity with noisy quantum channels.

%




\appendices

\section{Preliminaries on Quantum Information Theory} \label{append:qttheory}

In this section, we briefly introduce the fundamental framework of quantum information theory.
For more detail, see \cite{Hay17, Wilde, Tom}.

In classical information theory, 
the information is defined by an element $x$ of a finite set $\mathcal{X}$, and the information $x$ is changed by a function $f:\mathcal{X}\to\mathcal{Y}$, where $\mathcal{Y}$ is a finite set.
Similarly, in quantum information theory,
the quantum information is defined by a quantum state $\rho$ on a quantum system $\mathcal{A}$, and the quantum states $\rho$ on $\mathcal{A}$ is modified by quantum operations $\kappa$ from the states on $\mathcal{A}$ to the states on a quantum system $\mathcal{B}$.
Another difference between the two information theories is the measurement of information.
If there is no error on the measuring apparatus, 
the measurement of classical information $x$ is deterministic and does not change the information, 
i.e., measurement outcome is $x$ and the information $x$ is not changed after the measurement.
However, the measurement of a quantum state outputs its outcome probabilistically and changes the state.
In the following, we define the quantum system, quantum state, quantum operation, and quantum measurement.

A quantum system is defined by a finite-dimensional Hilbert space $\mathcal{A}$.
A vector $x$ in $\mathcal{A}$ is denoted by $|x\rangle$, and $\bar{x}^{\top}$ is denoted by $\langle x|$.
A quantum state is defined by a {\em density matrix} which is a positive semidefinite matrix $\rho$ on $\mathcal{A}$ such that $\Tr \rho = 1$.
We denote the set of density matrices on $\mathcal{A}$ by $\mathcal{S}(\mathcal{A})$.
When a state $\rho$ is rank-one, i.e., $\rho = |\psi\rangle\langle\psi|$ for some unit vector $|\psi\rangle$,  
the state is called a {\em pure state} and is identified with the vector $|\psi\rangle \in \mathcal{A}$.
If a state $\rho$ is not a pure state, it is called a {\em mixed state}.
The composite system of $\mathcal{A}$ and $\mathcal{B}$ is defined by $\mathcal{A}\otimes \mathcal{B}$.
For any quantum state $\rho\in\mathcal{S}(\mathcal{A}\otimes \mathcal{B})$,
the {\em reduced state} on $\mathcal{A}$ is described by $\Tr_\mathcal{B} \rho$, where $\Tr_\mathcal{B}$ is the partial trace on the system $\mathcal{B}$.
A state $\rho\in\mathcal{S}(\mathcal{A}\otimes \mathcal{B})$ is called a {\em separable state} if 
$\rho$ is written as $\rho = \sum_i p_i \sigma_i \otimes \tau_i$ for states $\sigma_i\in \mathcal{S}(\mathcal{A})$, $\tau_i\in\mathcal{S}(\mathcal{B})$, and a probability distribution $\{p_i\}_i$.
A state $\rho\in\mathcal{S}(\mathcal{A}\otimes \mathcal{B})$ is called an {\em entangled state} if $\rho$ is not separable.

A quantum operation is defined by a {\em Completely Positive Trace-Preserving (CPTP) linear map}
from $\mathcal{S}(\mathcal{A})$ to $\mathcal{S}(\mathcal{B})$.
A linear map $\kappa$ is called a {\em positive map} 
if $\kappa$ maps a positive semidefinite matrix to a positive semidefinite matrix,
and
is called a {\em Completely-Positive (CP) map} 
if the linear map $\kappa\otimes \iota_{\mathbb{C}^n}$ is a positive map for any positive integer $n$, where $\iota_{\mathbb{C}^n}$ is the identity map on $\mathcal{S}(\mathbb{C}^n)$.
An example of quantum operations is 
$\kappa_{\mathsf{U}}(\rho) \coloneqq \mathsf{U}\rho \mathsf{U}^{\dagger}$ for a unitary matrix $\mathsf{U}$ on $\mathcal{A}$.
By the operation $\kappa_{\mathsf{U}}$, a pure state $|\psi\rangle\in\mathcal{A}$ is mapped to the pure state $\mathsf{U} |\psi\rangle\in\mathcal{A}$.


A quantum measurement is defined by an {\em instrument}.
A set $\{\kappa_\omega\}_{\omega\in\Omega}$ of CP maps from $\mathcal{S}(\mathcal{A})$ to $\mathcal{S}(\mathcal{B})$ is called an {\em instrument} if 
for any quantum state $\rho\in\mathcal{S}(\mathcal{A})$, 
$$\sum_{\omega\in\Omega} \Tr \kappa_\omega (\rho) = 1.$$
With probability $\Tr \kappa_\omega (\rho)$, the measurement outcome is $\omega$ and the state after the measurement is 
${\kappa_\omega (\rho)}/{\Tr\kappa_\omega (\rho)}.$
When 
one is interested only in 
the measurement probability and the outcome, 
the measurement is described by a {\em Positive Operator-Valued Measure (POVM)}.
A set $\{M_\omega\}_{\omega\in\Omega}$ of positive semidefinite matrices is called an {\em POVM} if $\sum_{\omega\in\Omega} M_\omega = I$.
With probability $\Tr \rho M_\omega$, the measurement outcome is $\omega$.

\section{Quantum Information Measures} \label{append:entropy}

In this section, we introduce quantum information measures necessary for the analysis of QPIR protocols.

Any quantum state $\rho$ is diagonalized as $\rho = \sum_i p_i |i\rangle\langle i|$ for 
	a probability distribution $\{p_i\}_i$.
For a state $\rho = \sum_i p_i |i\rangle\langle i|$,
{\em von Neumann entropy} is defined by 
\begin{align}
H(\rho) &\coloneqq  H(\{p_i\}), \label{vonEn}
\end{align}
where $H(\cdot)$ in the right-hand side of \eqref{vonEn} is Shannon entropy $H(\{p_i\})\coloneqq -\sum_i p_i\log p_i$.
For any state $\rho\in\mathcal{S}(\mathcal{A} \otimes \mathcal{B})$,
we use the notation 
\begin{gather*}
H(\mathcal{A})_\rho \coloneqq H(\Tr_\mathcal{B} \rho), \quad
H(\mathcal{B})_\rho \coloneqq H(\Tr_\mathcal{A} \rho), \\
H(\mathcal{A},\mathcal{B})_\rho \coloneqq H(\rho).
\end{gather*}

For any state $\rho\in\mathcal{S}(\mathcal{A} \otimes \mathcal{B} \otimes \mathcal{C})$,
the {\em quantum conditional entropy}, {\em quantum mutual information}, and {\em quantum conditional mutual information} are defined as 
\begin{align}
H(\mathcal{A}|\mathcal{B})_\rho  &\coloneqq  H(\mathcal{A},\mathcal{B})_\rho - H(\mathcal{B})_\rho ,\\
I(\mathcal{A};\mathcal{B})_\rho &\coloneqq H(\mathcal{A})_\rho + H(\mathcal{B})_\rho - H(\mathcal{A}, \mathcal{B})_\rho,   \\
I(\mathcal{A};\mathcal{B} |\mathcal{C})_\rho &\coloneqq I(\mathcal{A};\mathcal{B},\mathcal{C})_{\rho}- I(\mathcal{A};\mathcal{C})_{\rho}.
\end{align}

When a quantum state $\rho_X \in \mathcal{S}(\mathcal{A})$ is prepared depending on the random variable $X\in\mathcal{X}$, 
the state on the composite system of $\mathcal{A}$ and $X$ is defined by
$$\tilde{\rho} =\sum_{x\in\mathcal{X}} \Pr[X=x]\cdot  \rho_x \otimes |x\rangle\langle x|.$$
For convenience, we denote 
$H(\cdot)_{\rho_{X}}\coloneqq H(\cdot)_{\tilde{\rho}}$ and
$I(\cdot)_{\rho_{X}}\coloneqq I(\cdot)_{\tilde{\rho}}$.

\section{Derivation of \eqref{eq:feff}} 
	\label{append:avvvv}

For the derivation of \eqref{eq:feff}, we use the data-processing inequality of quantum relative R\'{e}nyi entropy.
When $s\in(0,1)$,
	quantum relative R\'{e}nyi entropy is defined as 
\begin{align}
D_{1+s}(\rho \| \sigma) \coloneqq \frac{1}{s} \Tr \rho^{1+s} \sigma^{-s}
\end{align}
for any states $\rho$ and $\sigma$ such that $\mathrm{supp}(\rho) \subset \mathrm{supp}(\sigma)$,
and 
$D_{1+s}$ satisfies the data-processing inequality with respect to measurements:
\begin{align}
D_{1+s}(\rho \| \sigma) \ge D_{1+s}(P_\rho^{\mathcal{M}} \| P_\sigma^{\mathcal{M}}),
	\label{eq00}
\end{align}
where $P_{\rho}^{\mathcal{M}}$ and $P_{\rho}^{\mathcal{M}}$ are probability distributions after the measurement $\mathcal{M} = \{ M_i \}_i$ on $\rho$ and $\sigma$, respectively, i.e., 
\begin{align*}
P_{\rho}^{\mathcal{M}} = \sum_i  (\Tr\rho M_i)\cdot   |i\rangle \langle i |,  \quad
P_{\sigma}^{\mathcal{M}} = \sum_i  (\Tr\sigma M_i) \cdot |i\rangle \langle i | .
\end{align*}

Next, we prepare the following notations:
	\begin{align*}
	\sigma_{z} &\coloneqq \frac{1}{\vM} \sum_{w=0}^{\vM-1} \rho_{w,z},\\
	\tilde{\rho}_z &\coloneqq \frac{1}{\vM}\sum_{w=0}^{\vM-1} |w\rangle\langle w| \otimes \rho_{w,z},\\
	\tilde{\sigma}_z &\coloneqq \frac{1}{\vM}\sum_{w=0}^{\vM-1} |w\rangle\langle w| \otimes \sigma_{z},\\
	\tilde{\mathsf{Y}} &\coloneqq \sum_{w=0}^{\vM-1} |w\rangle\langle w| \otimes \mathsf{Y}_w,\\
	\mathcal{M} &= \{ \tilde{\mathsf{Y}} , \mathsf{I} - \tilde{\mathsf{Y}} \},
	\end{align*}
where $\{\mathsf{Y}_{w}\}_{w=0}^{\vM}$ is the decoding measurement defined in Section~\ref{sec:problem}.
With these notations, we have
\begin{align}
\Tr \tilde{\rho}_z \tilde{\mathsf{Y}} 
	&=  \frac{1}{\vM} \sum_{w=0}^{\vM-1} \Tr \rho_{w,z} \mathsf{Y}_w 
	= 1 - P_{\mathrm{err},z}(\Psi_{\mathrm{QPIR}}^{(\vM)}),
	\label{eq11}
	\\
\Tr \tilde{\sigma}_z \tilde{\mathsf{Y}} 
	&= \frac{1}{\vM} \sum_{w=0}^{\vM-1} \Tr \sigma_z \mathsf{Y}_w 
	\le \frac{1}{\vM} \Tr \sigma_z \sum_{w=0}^{\vM} \mathsf{Y}_w
	=  \frac{1}{\vM}.
	\label{eq22}
\end{align}

Combining \eqref{eq00}, \eqref{eq11}, and \eqref{eq22},
	we can derive Eq.~\eqref{eq:feff} similarly as \cite[(4.66)]{Hay17}:
\begin{align*}
 &(1-P_{\mathrm{err},z}(\Psi_{\mathrm{QPIR}}^{(\vM)}) )^{1+s} \vM^{s}\\
 &\stackrel{\mathclap{(a)}}{\le} 
	(\Tr \tilde{\rho}_z \tilde{\mathsf{Y}} )^{1+s} 
	(\Tr \tilde{\sigma}_z \tilde{\mathsf{Y}})^{-s}\\
 &\le (\Tr \tilde{\rho}_z \tilde{\mathsf{Y}} )^{1+s} 
	(\Tr \tilde{\sigma}_z \tilde{\mathsf{Y}})^{-s}\\
	&\quad +
	(1-\Tr \tilde{\rho}_z \tilde{\mathsf{Y}} )^{1+s} 
	(1-\Tr \tilde{\sigma}_z \tilde{\mathsf{Y}})^{-s}\\
 &= \exp \paren*{ sD_{1+s}( P_{\tilde{\rho}_z}^{\mathcal{M}} \| P_{\tilde{\sigma}_z}^{\mathcal{M}} ) }\\
 &\stackrel{\mathclap{(b)}}{\le} \exp \paren*{ sD_{1+s}( \tilde{\rho}_z \| \tilde{\sigma}_z ) }\\
 &= \Tr \tilde{\rho}_z^{1+s}  \tilde{\sigma}_z^{-s}
  = \frac{1}{\vM} \sum_{w=0}^{\vM-1} \Tr \rho_{w,z}^{1+s}  \sigma_{z}^{-s},
\end{align*}
where $(a)$ is from \eqref{eq11} and \eqref{eq22}
and $(b)$ is from \eqref{eq00}.

\section{Proof of Proposition~\ref{prop:cqfeed}}	\label{append:cqfeed}

For the proof of Proposition~\ref{prop:cqfeed}, we follow the proof of \cite[Theorem 4]{DQSW19}.
Before the proof, we prepare two lemmas from \cite{DQSW19}.
\begin{lemm}[{\cite[Lemma 2]{DQSW19}}] \label{qff0}
Let $\tau_{WFAB}$ be a classical-quantum state such that 
	\begin{align}
	\tau_{WFAB} = \sum_{w,f} p(w,f) |w,f\rangle \langle w,f| \otimes \tau_{AB|wf},
	\label{afvsdf}
	\end{align}
	where $\tau_{AB|wf}$ are pure states.
Let $\mathcal{M}$ be one-way Local Operations and Classical Communication (LOCC) map from $A\otimes B$ to $A'\otimes B' \otimes X$, {where $X$ is a classical system which is sent from $B$ to $A$.}
Then, we have
\begin{align}
	&I(W; B' F X) + H( B' | W F X) \\
	&\le I(W; B F) + H( B | W F) .
\end{align}
\end{lemm}

\begin{lemm}[{\cite[Lemma 3]{DQSW19}}] \label{qff}
Let $\tau_{WFAB}$ be a classical-quantum state defined in \eqref{afvsdf}.
Then 
	\begin{align}
	 &I(W; ABF ) + H(AB | W F )\\ 
	&\le H(A) + I(W; BF) + H(B | W F )     .
	\end{align}
\end{lemm}

For the proof, we formally describe the communication protocol as follows.
We denote the local registers of the sender and the receiver before the communication by $\mathcal{B}^0$ and $\mathcal{C}^0$. Let $\mathcal{A}^{0} = \mathbb{C}$.
At round $i\in\{1,\ldots, \vR\}$, the receiver applies a quantum instrument from $\mathcal{A}^{i-1}\otimes \mathcal{C}^{i-1}$ to $\mathcal{C}^{i}$ depending on $Q^{[i-1]}$ and sends the measurement outcome $Q^i$ to the sender
Then, the sender applies a quantum operation from $\mathcal{B}^{i-1}$ to $\mathcal{A}^i \otimes \mathcal{B}^i$ depending on $W$ and $Q^{[i]}\coloneqq (Q^{1},\ldots, Q^{i})$, and sends $\mathcal{A}^i$ to the receiver.
After the final $\vR$th-round, the sender applies a POVM on $\mathcal{A}^{\vR} \otimes \mathcal{C}^{\vR}$ depending on $Q^{[\vR]}$ and the measurement outcome $M$ is the decoding output.

%
%

Now, we prove Proposition \ref{prop:cqfeed}.
First, we have
\begin{align}
(1-\varepsilon) \log \vM 
	&\stackrel{\mathclap{(a)}}{\leq}
	I(W;M) + h_2(\varepsilon) \\
	&\stackrel{\mathclap{(b)}}{\le} I(W;\cA^{\vR} \mathcal{C}^{\vR} Q^{[\vR]}) + h_2(\varepsilon) ,
\end{align}
where $(a)$ is from Fano's inequality
	\begin{align}
	H(W|M) \le \varepsilon \log \vM + h_2(\varepsilon)
	\end{align}
	and the uniform distribution of $W$,
and $(b)$ is from the data-processing inequality for the decoding POVM.
Then, it is enough to derive the inequality
\begin{align}
 I(W;\cA^{\vR} \mathcal{C}^{\vR} Q^{[\vR]})
	&\le \sum_{i=1}^{\vR} H(\rho_W^{\cA^{i}}) 
	\label{eqffefsd}
\end{align}
for the proof of Proposition~\ref{prop:cqfeed}.

	To derive \eqref{eqffefsd}, we apply Lemma~\ref{qff0} and  Lemma~\ref{qff} as follows.
	Note that there is no constraint in the size of local registers.
	Thus, without losing generality, we assume that the sender's and the receiver's local registers are sufficiently large that the joint state on the entire protocol is always written as pure states.
	Since the operations at each round can be considered as a one-way LOCC map, 
	we can apply Lemma~\ref{qff0} for $(W,F,A,B,A',B',X) \coloneqq (W, Q^{[i-1]}, \cB^{i-1}, \cA^{i-1} \otimes \mathcal{C}^{i-1}, \cA^{i}\otimes \cB^{i}, \mathcal{C}^{i}, Q^{i})$:
	\begin{align*}
	&I(W; \mathcal{C}^{i} Q^{[i]}) + H( \mathcal{C}^{i} | W Q^{[i]})\\
	&\le I(W; \cA^{i-1} \mathcal{C}^{i-1} Q^{[i-1]}) + H( \cA^{i-1} \mathcal{C}^{i-1} | W Q^{[i-1]}).
	\end{align*}
	Furthermore, applying  Lemma~\ref{qff} with $(W,F, A,B) \coloneqq (W, Q^{[i]}, \cA^{i}, \mathcal{C}^{i} )$.
	we have
	\begin{align}
	&I(W; \cA^{i} \mathcal{C}^{i} Q^{[i]} ) + H(\cA^{i}  \mathcal{C}^{i} | W Q^{[i]} )
		\nonumber \\ 
	&\le H(\cA^{i}) + I(W; \mathcal{C}^{i} Q^{[i]}) + H(\mathcal{C}^{i} | W Q^{[i]} )   . \label{eq1:vdvs}
	\end{align}
Combining the above two inequalities, we have
	\begin{align}
	&I(W; \cA^{i} \mathcal{C}^{i} Q^{[i]} ) + H(\cA^{i}  \mathcal{C}^{i} | W Q^{[i]} )
	\nonumber
	\\
	&\le H(\cA^{i}) + I(W; \cA^{{i}-1} \mathcal{C}^{{i}-1} Q^{[i-1]}) 
	\nonumber
	\\
	&\quad + H(\cA^{{i}-1}\mathcal{C}^{{i}-1} | W Q^{[i-1]} )  \label{eqfsdfd}
	\end{align}
Applying the inequality \eqref{eqfsdfd} recursively, we obtain the desired inequality \eqref{eqffefsd} as 
	\begin{align*}
	&I(W; \cA^{\vR} \mathcal{C}^{\vR} Q^{[\vR]})\\
	&\le I(W; \cA^{\vR} \mathcal{C}^{\vR} Q^{[\vR]} ) + H(\cA^{\vR}  \mathcal{C}^{\vR} | W Q^{[\vR]} )\\ 
	&\stackrel{\mathclap{(c)}}{\le} \sum_{i=2}^{\vR} H(\cA^{i}) + I(W; \cA^{1} \mathcal{C}^{1} Q^{1}) + H(\cA^{1}\mathcal{C}^{1} | W Q^{1} )  
		\\
	&\stackrel{\mathclap{(d)}}{\le} \sum_{i=1}^{\vR} H(\cA^{i}) + I(W; \mathcal{C}^{1} Q^{1}) + H(\mathcal{C}^{1} | W Q^{1} )     
		\\
	&\stackrel{\mathclap{(e)}}{=} \sum_{i=1}^{\vR} H(\cA^{i})     ,
	\end{align*}
	where 
	$(c)$ is derived by applying \eqref{eqfsdfd} recursively for $i=\vR,{\vR-1},\ldots, 2$,
	$(d)$ is from \eqref{eq1:vdvs}, and 
	$(e)$ is obtained as follows:
	$I(W;  \mathcal{C}^{1} Q^{1}) = 0 $ because the receiver prepares the state of $\mathcal{C}^{1} \otimes Q^{1}$ independently of the sender's message $W$,
	and
	$H(\mathcal{C}^{1} | W Q^{1} ) = 0 $ 
		since 
		the initial state of the local register $\mathcal{C}^1$ is a pure state. 


\begin{IEEEbiographynophoto}{Seunghoan Song}(S'20)
received the B.E. degree from Osaka University in 2017 and the Master of Mathematical Science degree from Nagoya University in 2019. 
He is currently pursuing Ph.D. degree at the Graduate School of Mathematics, Nagoya University.
He is also a Research Fellow of the Japan Society of the Promotion of Science (JSPS) from 2020.
He awarded the School of Engineering Science Outstanding Student Award in 2017 and 
Graduate School of Mathematics Award for Outstanding Masters Thesis in 2019.
His research interests include classical and quantum information theory and its applications to secure communication protocols.
\end{IEEEbiographynophoto}

\begin{IEEEbiographynophoto}{Masahito Hayashi}(M'06--SM'13--F'17) was born in Japan in 1971.
He received the B.S.\ degree from the Faculty of Sciences in Kyoto
University, Japan, in 1994 and the M.S.\ and Ph.D.\ degrees in Mathematics from
Kyoto University, Japan, in 1996 and 1999, respectively. He worked in Kyoto University as a Research Fellow of the Japan Society of the
Promotion of Science (JSPS) from 1998 to 2000,
and worked in the Laboratory for Mathematical Neuroscience,
Brain Science Institute, RIKEN from 2000 to 2003,
and worked in ERATO Quantum Computation and Information Project,
Japan Science and Technology Agency (JST) as the Research Head from 2000 to 2006.
He also worked in the Superrobust Computation Project Information Science and Technology Strategic Core (21st Century COE by MEXT) Graduate School of Information Science and Technology, The University of Tokyo as Adjunct Associate Professor from 2004 to 2007.
He worked in the Graduate School of Information Sciences, Tohoku University as Associate Professor from 2007 to 2012.
In 2012, he joined the Graduate School of Mathematics, Nagoya University as Professor.
Also, he was appointed in Centre for Quantum Technologies, National University of Singapore as Visiting Research Associate Professor from 2009 to 2012
and as Visiting Research Professor from 2012 to now.
He worked in Center for Advanced Intelligence Project, RIKEN as
a Visiting Scientist from 2017 to 2020.
He worked in Shenzhen Institute for Quantum Science and Engineering, Southern University of Science and Technology, Shenzhen, China as a Visiting Professor from 2018 to 2020,
and
in Center for Quantum Computing, Peng Cheng Laboratory, Shenzhen, China
as a Visiting Professor from 2019 to 2020.
In 2020, he joined Shenzhen Institute for Quantum Science and Engineering, Southern University of Science and Technology, Shenzhen, China
as Chief Research Scientist. 
In 2011, he received Information Theory Society Paper Award (2011) for ``Information-Spectrum Approach to Second-Order Coding Rate in Channel Coding''.
In 2016, he received the Japan Academy Medal from the Japan Academy
and the JSPS Prize from Japan Society for the Promotion of Science.

In 2006, he published the book ``Quantum Information: An Introduction'' from Springer, whose revised version was published as ``Quantum Information Theory: Mathematical Foundation'' from Graduate Texts in Physics, Springer in 2016.
In 2016, he published other two books ``Group Representation for Quantum Theory'' and ``A Group Theoretic Approach to Quantum Information'' from Springer.
He is on the Editorial Board of {\it International Journal of Quantum Information}
and {\it International Journal On Advances in Security}.
His research interests include classical and quantum information theory and classical and quantum statistical inference.
\end{IEEEbiographynophoto}

\end{document}